\documentclass[reqno,centertags, 12pt, draft]{article}

\usepackage{amscd}

\usepackage{mathrsfs}

\usepackage{amsmath,amssymb, amsfonts, amsthm, latexsym,titlesec,tikz}

\newtheorem{theorem}{Theorem}[section]

\newtheorem{lem}{Lemma}[section]
\newtheorem{lemma}[lem]{Lemma}
\newtheorem{cor}[lem]{Corollary}
\newtheorem{propos}[lem]{Proposition}

\newtheorem{rem}[lem]{Remark}

\newtheorem{defin}{Definition}[section]

\newcommand{\N}{\mathbb{N}}
\newcommand{\E}{\mathbb E}
\newcommand{\Z}{\mathbb{Z}}
\newcommand{\R}{\mathbb{R}}

\newcommand{\Ex}[1]{\mathbb E\left[#1\right]}
\newcommand{\Tr}[1]{\textnormal{Tr}\left(#1\right)}
\newcommand{\Var}[1]{\textnormal{Var}\left(#1\right)}
\newcommand{\Cov}[1]{\textnormal{Cov}\left(#1\right)}

\numberwithin{equation}{section}

\begin{document}

\title{Spectral fluctuations for Schr\"odinger operators with a random decaying potential}
\author{Jonathan Breuer, Yoel Grinshpon, Moshe White \thanks{Institute of Mathematics, The Hebrew University of Jerusalem, Jerusalem, 91904, Israel.
Emails: jbreuer@math.huji.ac.il, yoel.grinshpon@mail.huji.ac.il, moshe.white@mail.huji.ac.il}}

\maketitle

\begin{abstract}
We study fluctuations of polynomial linear statistics for discrete Schr\"odinger operators with a random decaying potential. We describe a decomposition of the space of polynomials into a direct sum of three subspaces determining the growth rate of the variance of the corresponding linear statistic. In particular, each one of these subspaces defines a unique critical value for the decay-rate exponent, above which the random variable has a limit that is sensitive to the underlying distribution and below which the random variable has asymptotically Gaussian fluctuations.
\end{abstract}

\sloppy

\section{Introduction}

The purpose of this paper is to study fluctuations of the eigenvalue counting measure for discrete half-line Schr\"odinger operators with a random decaying potential. To be precise, we consider the operator
\begin{equation} \nonumber
(H_{\alpha}u)_n= \left\{ \begin{array}{cc} u_{n-1}+u_{n+1}+V_{\alpha}(n)\cdot u_n & n>1  \\
u_2+V_{\alpha}(1)\cdot u_1 &  n=1 \end{array} \right.
\end{equation}
on $\ell^2(\N)$, where $0<\alpha$ and $V_{\alpha}(n)=\frac{X_n}{n^\alpha}$ with $\{X_n\}_{n=1}^\infty$ a sequence of independent, identically distributed (iid) random variables satisfying $\Ex{X_n}=0$. We denote by $\eta^2=\Var{X_n} =\Ex{X_n^2}>0$. For the simplicity of some of the arguments below, we also assume that there exists some $C>0$ such that $|X_n|<C$ almost surely.

Such operators have received a fair amount of attention in the '80s and '90s of the previous century \cite{DSS2,DSS,KLS,Simon}, since they exhibit an analyzable spectral transition (from pure point spectrum to absolutely continuous spectrum via singular continuous spectrum), which occurs as the rate of decay, $\alpha$, increases. In particular, for $\alpha<1/2$ such operators are known to have pure point spectrum with super-polynomially decaying eigenfunctions with probability $1$, \cite{Simon,DSS2, DSS} whereas for $\alpha>1/2$ the spectrum is almost surely absolutely continuous \cite{KLS}. In the critical case of $\alpha=1/2$ the spectral measure may have a singular continuous as well as a pure point component, depending on $\eta$ \cite{KLS}.

The past decade has seen renewed interest in these operators and related models, due to the rich asymptotic spacing properties of eigenvalues of their finite volume approximations \cite{BLS, Kil-St,  KN, KN2, KVV}. Roughly speaking, these asymptotics undergo a qualitative transition from local `Poisson statistics'\footnote{Strictly speaking, Poisson statistics have so far only been shown for the continuum model \cite{KN2} and for the unit circle analog \cite{Kil-St}.} to `clock behavior' corresponding to the transition from pure point spectrum to absolutely continuous spectrum \cite{BLS,KN}. Moreover, at the critical decay rate of $n^{-1/2}$, where the transition occurs, the local eigenvalue statistics have random matrix asymptotics \cite{KVV,Kil-St, KN}.

Our aim in this paper is to analyze the fluctuations of the counting measure for the eigenvalues of finite volume approximations. Explicitly, we study fluctuations of polynomial linear statistics of the finite volume ensemble with the purpose of analyzing the rate of growth of the variance and understanding when the limit is Gaussian. Research into problems of this type has been common in models of random matrix theory (there is a huge number of works on fluctuations of linear statistics in this context; \cite{serfaty, Bor, BG, BGG, BD, BG, DE, DP, J1, J2, lambert, sosh1, sosh2} is a very partial list of relevant references) and lately also in the context of the Anderson model \cite{KP,PS}. The growth rate of the variance is considered to be an indicator of the rigidity of the eigenvalue locations. Fluctuations of the integrated density of states in the analogous continuum model have been studied by Nakano in \cite{Nakano}.

Remarkably, as we show, in the case of the Schr\"odinger operator with a random decaying potential, the underlying symmetry of the operator leads to a decomposition of the space of polynomials into subspaces with different asymptotic fluctuations. In particular, the space of polynomials decomposes naturally as a direct sum of three spaces, each of which has a different variance growth rate, and each of which has Gaussian asymptotics for different values of the relevant parameters.

We note that Jacobi matrices have also been utilized in the context of random matrix theory to study asymptotic fluctuations of the eigenvalue point process \cite{BD, DE1, DE, Popescu}. Indeed, Dumitriu-Edelman \cite{DE} and Popescu \cite{Popescu} study polynomials of random Jacobi matrix models with \emph{growing} off diagonal entries that arise naturally in this context. While such operators exhibit a spectral transition similar to the decaying random potential case \cite{Breuer-TAMS,BFS}, the asymptotic fluctuations are very different. We compare our results with theirs below towards the end of the Introduction.

In order to present our results we need some definitions. Given $N \in \N$, let
$$
H_{\alpha,N}=\begin{pmatrix}
V_{\alpha}\left(1\right) & 1 & 0 & \ldots & 0\\
1 & V_{\alpha}\left(2\right) & \ddots & \ddots & \vdots\\
0 & \ddots & \ddots & \ddots & 0\\
\vdots & \ddots & \ddots & \ddots & 1\\
0 & \ldots & 0 & 1 & V_{\alpha}\left(N\right)
\end{pmatrix}
$$
(the restriction of $H_\alpha$ to $[1,N]$). The \emph{empirical measure} of $H_{\alpha,N}$ is the measure
$$\textrm{d}\nu_{\alpha,N}=\frac{1}{N}\sum_{i=0}^{N} \delta_{\lambda_i^{(N)}}$$
where $\left \{\lambda_{1}^{(N)},\lambda_2^{(N)},\ldots,\lambda_N^{(N)} \right \}=\textrm{sp}\left(H_{\alpha,N} \right)$ is the spectrum of $H_{\alpha,N}$, and $\delta_{\lambda_i^{(N)}}$ is the Dirac measure at $\lambda_i^{(N)}$. When the empirical measure has a limit as $N \rightarrow \infty$, this limit is known as the density of states (DOS) of $H$. In the case of $H_{\alpha}$, since the operator is a compact perturbation of the operator with $V_n \equiv 0$, the random measure $\textrm{d}\nu_{\alpha,N}$ converges weakly almost surely to the arcsine measure:
$$\textrm{d}\nu(x)=\frac{1}{\pi\sqrt{4-x^2}}\chi_{[-2,2]}(x)\textrm{d}x.$$

We want to focus on asymptotics of the fluctuations of $\textrm{d}\nu_{\alpha,N}$. A natural way to study this is using \emph{linear statistics}, namely random variables of the form $\int f \textrm{d}\nu_N=\frac{1}{N}\Tr{f(H_{\alpha,N})}$ for some function $f$ defined on the spectrum of $H_{\alpha,N}$. We shall study these fluctuations for polynomial $f$'s.

For an odd and positive integer $m$, let $\mathcal{V}_m$ be the space of polynomials over $\mathbb{R}$ of degree $\le m$ with the inner product
\begin{equation} \nonumber
\left\langle \sum_{k=0}^m a_k x^k, \sum_{k=0}^m b_k x^k \right \rangle = \sum_{k=0}^m a_k b_k.
\end{equation}

Let $\mathcal{E}_m \subseteq \mathcal{V}_m$ be the subspace of polynomials in $x^2$, and let $\mathcal{O}_m=x \mathcal{E}_m$ be its orthogonal complement (i.e.\ the subspace of polynomials in odd powers of $x$). In $\mathcal{O}_m$, let $\mathcal{Q}_m$ be the one-dimensional subspace spanned by the polynomial
\begin{equation} \label{Qdef}
Q_m(x)=\sum_{j=0}^{(m-1)/2} \frac{(2j+1)!}{\left(j \right)!^2} x^{2j+1}
\end{equation}
and $\mathcal{Q}_m^\perp \subseteq \mathcal{O}_m$ its orthogonal complement in $\mathcal{O}_m$. We have
\begin{equation} \label{eq:decomp}
\mathcal{V}_m=\mathcal{Q}_m\oplus \mathcal{Q}_m^\perp\oplus \mathcal{E}_m.
\end{equation}

We denote by $N\left(0,\sigma^2\right)$ the Normal distribution on $\R$ with mean $0$ and variance $\sigma^2$, and denote by $\overset{d}{\longrightarrow}$ convergence in distribution.
Finally, for $0<t\leq 1$, denote
$$g_t(N)^2=\begin{cases}
\log N=\int_1^N \frac{1}{x^t}dx & t=1\\
{\frac{N^{1-t}}{1-t}}=\int_1^N \frac{1}{x^t}dx+\frac{1}{{1-t}} & 0<t<1
\end{cases},$$
and note that
\begin{equation} \label{eq:AysmptoticSum}
\frac{g_t(N)^2}{\sum_{n=1}^N \frac{1}{n^t}}\overset{N\longrightarrow \infty}{\longrightarrow} 1.
\end{equation}

Our main result says that for each of the different subspaces in the decomposition \eqref{eq:decomp} there is a critical value $\alpha_c$ such that for $\alpha\leq \alpha_c$, $\frac{\Tr{P(H_{\alpha,N})}-\Ex{\Tr{P(H_{\alpha,N})}}}{g_{\alpha/\alpha_c}(N)}$ converges to a Gaussian random variable, while for $\alpha>\alpha_c$, $\Tr{P(H_{\alpha,N})}-\Ex{\Tr{P(H_{\alpha,N})}}$ converges to a random variable (whose distribution depends on the distribution of the $X_n$'s), for any $P$ in that subspace. Explicitly,

\begin{theorem} \label{thm:main} %FIXED SPACING
Let $P \in \mathcal{V}_m$ be a polynomial with $\deg(P)>0$. \\
\begin{verse}
a) (i) For $0<\alpha \le \alpha_c^{\mathcal{Q}}:=\frac12$,
  $$\frac{\Tr{P(H_{\alpha,N})}-\Ex{\Tr{P\left(H_{\alpha,N}\right)}}}{g_{2\alpha}(N)}\overset{d}{\longrightarrow} N\left(0,\sigma_{\mathcal{Q}}(P)^2 \right) $$
  as $N \rightarrow \infty$, where
\begin{equation} \nonumber
\sigma_{\mathcal{Q}}(P)^2=\left|\left \langle P,Q_m \right \rangle \right|^2 \eta^2.
\end{equation} \\

\ \ \ \ (ii) For $\alpha>\frac12$, $\Tr{P(H_{\alpha,N})}-\Ex{\Tr{P\left(H_{\alpha,N} \right)}}$ converges a.s.\ to a random variable with finite variance. \\

\medskip

b) Assume $P \in \mathcal{E}_m$.

\ \ \ \ (i) For $\alpha \leq \alpha_c^{\mathcal{E}}:= \frac14$,
$$\frac{\Tr{P(H_{\alpha,N})}-\Ex{\Tr{P(H_{\alpha,N})}}}{g_{4\alpha}(N)}\overset{d}{\longrightarrow} N\left(0,\sigma_{\mathcal{E}}(P)^2 \right) $$
  as $N \rightarrow \infty$. If $\Var{X_n^2}> 0$ or $\deg(P)>2$, then $\sigma_{\mathcal{E}}(P)^2>0$. \\%Changed \neq to > in Var.

\ \ \ \ (ii) For $\alpha>\frac14$, $\Tr{P(H_{\alpha,N})}-\Ex{\Tr{P(H_{\alpha,N})}}$ converges a.s.\ to a random variable with finite variance. \\

\medskip

c) Assume $P \in \mathcal{Q}_m^\perp$.

\ \ \ \ (i) For $\alpha \leq \alpha_c^{\mathcal{Q}^\perp}:=\frac16$,
$$\frac{\Tr{P(H_{\alpha,N})}-\Ex{\Tr{P(H_{\alpha,N})}}}{g_{6\alpha}(N)}\overset{d}{\longrightarrow} N\left(0,\sigma_{\mathcal{Q}^\perp}(P)^2 \right) $$
 as $N \rightarrow \infty$. We have $\sigma_{\mathcal{Q}^\perp}(P)^2>0$, unless both $\Var{X_n^2}=0$ and $P(x)=a\left(\frac{x^5}{5}-4x^3+18x\right)$, for some constant $a$. \\ % Should we multiply by 5?

\ \ \ \ (ii) For $\alpha>\frac16$, $\Tr{P(H_{\alpha,N})}-\Ex{\Tr{P(H_{\alpha,N})}}$ converges a.s.\ to a random variable with finite variance.

\end{verse}
\end{theorem}

\begin{rem}
We compute explicit formulas for $\sigma_{\mathcal{E}}(P)$ $($see Proposition \ref{propos:even}$)$ and for $\sigma_{\mathcal{Q}^\perp}(P)$ $($see Proposition \ref{propos:perp}$)$ below. As they are quite cumbersome to describe we do not present them here. It takes a separate argument to show that they are positive under the conditions described above.
\end{rem}

\begin{rem}
Suppose that the random variables $X_n$ only take the values $c$ and $-c$ for some $c\neq 0$. For $P(x)=ax^2+b$ for some $a,b$, it is clear in this case that $\Tr{P(H_{\alpha,N})}-\Ex{\Tr{P(H_{\alpha,N})}}\equiv 0$.\\
Similarly, if $P(x)=a\left(\frac{x^5}{5}-4x^3+18x\right)$ for some $a$, we show in the proof of Proposition \ref{propos:perp} below that $\Var{\Tr{P(H_{\alpha,N})}-\Ex{\Tr{P(H_{\alpha,N})}}}=o(g_{6\alpha}(N))$.
%For $P(x)=x^2$ and random variables $X_n$ taking values $a$ or $-a$ with probability $\frac12$ each, $\Tr{P(H_{\alpha,N})}-\Ex{\Tr{P(H_{\alpha,N})}}\equiv 0$.

\end{rem} % APPROVE CHANGE

\begin{rem}
As noted above, in the case of $\alpha>\alpha_c$ the limit of $\Tr{P(H_{\alpha,N})}-\Ex{\Tr{P(H_{\alpha,N})}}$ $($which we show exists$)$ is not necessarily Gaussian and in fact depends on the distribution of the $X_n$'s. For example, for $P(x)=x$, if the $\left \{X_n \right \}$ are uniformly distributed on the interval $[-1,1]$, then the characteristic function of $X_n$, denoted by $\phi_{X_n}(t)$, is $\frac{\sin t}{t}$, and $\phi_{\left(V_\alpha(n) \right)}(t)=\frac{\sin\left(\frac{t}{n^{\alpha}}\right)}{\frac{t}{n^{\alpha}}}$. Since $V_\alpha(n)$ are independent, the characteristic function of $\Tr{P(H_{\alpha,N})}=\sum_{n=1}^N V_\alpha(n)$ is $\prod_{n=1}^N \frac{\sin(\frac{t}{n^{\alpha}})}{\frac{t}{n^{\alpha}}}$, which does not converge to the characteristic function of a Gaussian distribution $($e.g.\ it vanishes at $2\pi)$.
\end{rem}

The proof of Theorem \ref{thm:main} proceeds by considering first the asymptotic fluctuations of monomials. The trace of $H_{\alpha,N}^k$ can be computed by counting paths on a graph described in Section 2 below. This is a well known fact in the context of Jacobi matrices that goes back at least to work of Flajolet \cite{flajolet} and Viennot \cite{viennot}, and has been used recently in exploiting the connection between Jacobi matrices and random matrix models (see, e.g., \cite{DE, Duy, HA, GS, Popescu, Wong}). More generally, casting the problem of computing moments as a path-enumeration problem is a classical method in random matrix theory \cite{RMT}. To the best of our knowledge, however, this approach has never been applied to the model considered here, where it leads to the surprising results described above. In our case, the decay of $V_\alpha(n)$ implies that only a small number of paths contribute significantly to the asymptotics and these can be computed explicitly.

The difference in the rate of growth between even and odd $k$ follows from the fact that the off diagonal elements of $H_\alpha$ are non-random. Letting
\begin{equation} \nonumber
(\Delta u)_n=\left\{ \begin{array}{cc} u_{n-1}+u_{n+1} & n>1  \\
u_2 & n=1  \end{array} \right.
\end{equation}
so that $H_\alpha=\Delta+V_{\alpha}$ (where we abuse notation and use $V_{\alpha}$ for the operator of multiplication by the function $V_{\alpha}$), it is not hard to see that the diagonal of $H_{\alpha,N}^k$ is composed of only odd (even) powers of $V_{\alpha}$ when $k$ is odd (even) (see Corollary \ref{oddeven} below). Thus, while the lowest power of $V_{\alpha}(n)$ ($1\leq n \leq N$) in $\Tr{H_{\alpha,N}^k}$ is $1$ in the case of odd $k$, it is $2$ for $k$ even, which implies a difference by a factor of $2$ in the decay rate. The condition for a polynomial, $P$, of being in $\mathcal{Q}_m^\perp$ means in fact that the summands of order $1$ in $\Tr{P\left(H_{\alpha,N} \right)}$ vanish. Since $\mathcal{Q}_m\subseteq \mathcal{O}_m$ this means that the lowest power of $V_{\alpha}$ in the diagonal of $P\left(H_{\alpha,N} \right)$ is $3$. This accounts for the special decay rate for $P \in \mathcal{Q}_m^\perp$.

In terms of spectral fluctuations, the different rates can be seen as coming from the following symmetry of $H_\alpha$: $\Delta$ is unitarily equivalent to $-\Delta$ through the unitary transformation $u_n \mapsto (-1)^nu_n$. This means that $H_\alpha=\Delta+V_{\alpha}$ is unitarily equivalent to $-H_{\alpha}^-=-(\Delta-V_{\alpha})$ and a similar equivalence holds for the truncated operator. This implies that for even $k$
\begin{equation} \label{eq:SignChange}
\textrm{sp}\left(H_{\alpha,N}^k \right)=\textrm{sp}\left(\left( H_{\alpha,N}^- \right)^k \right).
\end{equation}
Since $V_\alpha$  is decaying, the fluctuations of $\textrm{sp}(H_{\alpha,N})$ are determined by its first several values and by \eqref{eq:SignChange} fluctuations that reverse the sign of these values are insignificant for $\Tr{P\left(H_{\alpha,N} \right)}$ for $P \in \mathcal{E}_m$. However, such fluctuations are certainly significant for $P \in \mathcal O_m$. The different rates are thus also expected on spectral grounds.

Before we end this introduction, we comment on the related model of random Jacobi matrices with growing off diagonal entries where both the off diagonal elements and the diagonal elements are random. To be explicit, we want to briefly discuss matrices of the form
$$
J=\begin{pmatrix}
b_{1} & a_{1} & 0 & \ldots\\
a_{1} & b_{2} & a_{2} & \ddots\\
0 & a_{2} & \ddots & \ddots\\
\vdots & \ddots & \ddots & \ddots
\end{pmatrix}
$$%APPROVE (redid matrix, old version follows)
%$$ J=
%\begin{pmatrix}
%b_1 & a_1 & 0 & \ldots & \ldots & \ldots  \\
%a_1 & b_2 & a_2 & 0 & \ddots &\ddots \\
%0&\ddots&\ddots&\ddots&\ddots &\ddots\\
%\vdots&\ddots&\ddots&\ddots&\ddots&\ddots\\
%\vdots&\ddots&\ddots&\ddots&\ddots&\ddots\\
%\vdots&\ddots&\ddots&\ddots&\ddots&\ddots\\
%\end{pmatrix}
%$$
where $a_n=n^\eta+\lambda Y_n$ and $b_n=\lambda X_n$ with $\{X_n\}_{n=1}^\infty,\{Y_n\}_{n=1}^\infty$, two sequences of iid random variables with finite moments (the model studied in \cite{Breuer-TAMS} is more general, but we restrict our discussion to this case for simplicity). Continuity properties of the spectral measure for such matrices were studied in \cite{Breuer-TAMS}, where it was shown that these operators are susceptible to a similar analysis as that applied in the random decaying case and they exhibit a similar transition from pure point to continuous spectrum with $\eta$ playing the same role as $\alpha$ does for the decaying model\footnote{Note, however, that absolute continuity of the spectral measure for $\eta>1/2$ is still an open problem}. However, in contrast with the bounded (i.e.\ random decaying) case, the continuity of the spectral measure is uniform throughout the spectrum.

The original motivation for studying the model in \cite{Breuer-TAMS} comes from the fact that it is a natural generalization\footnote{Strictly speaking, the model arising in $\beta$-ensembles is not exactly a particular case of this general model, but is very close.} of a random Jacobi matrix model arising in the study of Gaussian $\beta$ ensembles \cite{DE1}. However, the asymptotic local eigenvalue process has not yet been studied for the model in \cite{Breuer-TAMS}. We expect its dependence on the parameters $\eta$ and $\lambda$ to be universal so that for $\eta<1/2$ we expect local asymptotic Poisson behavior, for $\eta>1/2$ we expect asymptotic clock behavior, and for $\eta=1/2$ we expect the local eigenvalue process to be asymptotic to a Gaussian $\beta$ ensemble with $\beta$ an explicit function of $\lambda$.

As for the fluctuations of polynomial linear statistics, these have been studied in \cite{DE} for the particular model arising from Gaussian $\beta$ ensembles and for more general matrices in this family in \cite{Popescu}. Since the matrices are unbounded, a rescaling is required in order to study asymptotic fluctuations of polynomials. The properly scaled truncated matrices are $A_N=\frac{J_N}{N^\eta}$ and for these matrices, Popescu \cite[Theorem 3]{Popescu} has shown that for any polynomial, $P$,
\begin{equation} \nonumber
N^{\eta-1/2}\left(\Tr{P(A_N)}-\Ex{\Tr{P(A_N)}}\right)
\end{equation}
converges to a Gaussian random variable with an explicit formula for the variance. Note that there is no need in this case to differentiate between subspaces of polynomials. This is due to the fact that the off diagonal elements are random here so the discussion in the previous paragraphs does not apply. Moreover, for any values of the parameter $\eta$ and with the proper scaling, one has Gaussian asymptotics. As for variance growth note that for $\eta>1/2$, the normalization $N^{\eta-1/2}$ implies that the variance of $\Tr{P(A_N)}$ is decaying with $N$, which corresponds to a rigidity of the eigenvalue process on the global scale. For $\eta<1/2$, the variance is growing, as one would expect if the eigenvalue process is indeed, at least locally, Poisson; and for $\eta=1/2$ the variance is bounded, which is well known to be the case for the $\beta$ ensembles (\cite{J2, DE1}).

We note again that in the decaying model, by Theorem \ref{thm:main}, for any polynomial, $P$, if $\alpha>1/2$ then $\Tr{P(H_{\alpha,N})}$ has a limit which is not necessarily Gaussian. Together with the fact that the variance growth rate depends on the subspace to which $P$ belongs, these are notable differences between the decaying diagonal and growing off diagonal models.

The rest of this paper is structured as follows. The next section contains a break-up of the proof of Theorem \ref{thm:main} into several propositions dealing with asymptotic fluctuations for monomials. In Section 3, the combinatorial identities needed for the computation of the leading terms in the asymptotics are formulated and proved, and in Section 4, Propositions \ref{propos:NoCLT}, \ref{propos:odd}, \ref{propos:even}, and \ref{propos:perp}, which compose the proof of Theorem \ref{thm:main}, are proved.

\textbf{Acknowledgments.} We are grateful to Ofer Zeitouni for a useful discussion.

Research by JB and YG was supported in part by the Israel Science Foundation (Grant No. 399/16) and in part by the United States-Israel Binational Science Foundation (Grant No. 2014337). Research by MW was supported in part by the Israel Science Foundation (Grant No. 1612/17) and in part by the United States-Israel Binational Science Foundation (Grant No. 2006066).

%%%%%%%%%%%%%%%%%%%%%%%%%%%%%%%%%%%%%%%%%%%%%%%%%%%%%%%%%%%%% Section 2 %%%%%%%%%%%%%%%%%%%%%%%%%%%%%%%%%%%%%%%%%%%%%

\section{Fluctuations of monomials and path counting}

Fix $k\in\mathbb{N}$ and note that $\Tr{H_{\alpha,N}^{k}}$ is a polynomial in the variables $V_1,\ldots,V_N$ (given $\alpha> 0$, we denote $V_n=V_\alpha(n)$ for simplicity). In order to prove Theorem \ref{thm:main}, we need to understand the asymptotic behavior (as $N\rightarrow\infty$) of coefficients of low degree monomials in this polynomial. We shall use multi-indices to label these coefficients.%DECIDE (THIS IS MEANT AS A COMPROMISE TO KEEP NOTATION SIMPLE IN SECTIONS 2-3 AND CONSISTENT IN SECTION 4 WHEN ALPHA VARIES).

\begin{defin}
A multi-index is defined as a finitely supported function $\beta:\Z\rightarrow\N \cup\left\{ 0\right\}$.
Let $\beta_{h}$ denote $\beta\left(h\right)$ for any $h\in\Z$.
If $\beta$ is supported on $\N$, we let $V^{\beta}$ denote
the monomial $\prod_{h}V_{h}^{\beta_{h}}$. The degree of the monomial
is $\sum_{h}\beta_{h}$, which we also denote by $\left|\beta\right|$.

For any multi-index $\beta$ and $i\in\mathbb{Z}$, define a new multi-index
$\beta^{i}$ by $\beta_{h}^{i}=\beta_{h-i}$ for every $h\in\mathbb{Z}$.
We also fix a multi-index $\delta$ by
$$\delta_{h}=\begin{cases}
1 & h=0\\
0 & h\neq0
\end{cases}.$$
Note that using this notation, $\delta_{h}^{i}$ is $1$ if $h=i$,
and $0$ otherwise.
\end{defin}

The key to the analysis is that the coefficient of a term of the form $V^\beta$ is given by a certain number of paths in a directed graph whose vertices are given by $\left(\mathbb{N}\cup\{0\}\right)\times\mathbb{Z}$ and directed edges are of the form $\left((x,y) \rightarrow (x+1,y+m)\right)$ where $m=-1,0,1$. A path of length $k$ in this graph, that originates at $(0,y_0)$, is uniquely associated with a sequence of integers $\left(y_{0},y_{1},\ldots,y_{k}\right)$, where $\left|y_{t}-y_{t-1}\right|\leq1$. Thus we shall refer to such sequences as paths as well. Following \cite{Popescu} we call each pair $\left(y_{t-1},y_{t}\right)$ in such a sequence a \emph{step}. A step $\left(y_{t-1},y_{t}\right)$ is called \textit{up}, \textit{down}, or \textit{flat}, if $y_{t}-y_{t-1}$ equals $+1,-1,0$ respectively. A step $\left(y_{t-1},y_{t}\right)$ is said to be at the \textit{level} $y_{t-1}$.

\begin{rem}
While $\Tr{H_{\alpha,N}^{k}}$ is a function of $V^\beta$ only for $\beta$ supported on $\mathbb{N}$, it is convenient for technical purposes to consider our multi-indices as defined over $\mathbb{Z}$. For the same reason we consider our paths in $\left(\mathbb{N}\cup\{0\}\right)\times\mathbb{Z}$ rather than over $\mathbb{N}\times \mathbb{N}$.
\end{rem}

Since we are dealing with the trace, we shall be mostly interested in paths that start and end at the same point. For a multi-index $\beta$, let $\mathcal{P}^k(\beta)$ be the set of paths $\left(0=y_{0},y_{1},\ldots,y_{k}=0\right)$ such that $\beta^{i}$ counts flat steps in the path, for some integer $i$. Namely, $\left(0=y_{0},y_{1},\ldots,y_{k}=0\right) \in \mathcal{P}^k(\beta)$ if there is some $i\in\mathbb{Z}$ so that for every $h\in\mathbb{Z}$, $\beta_{h-i}$ is the number of flat steps of $\left(y_{0},y_{1},\ldots,y_{k}\right)$ at level $h$ (see Figures 1 and 2). Let $p^{k}\left(\beta\right)$ denote the number of paths in $\mathcal{P}^k(\beta)$. Thus, for example, $p^{k}\left( \delta \right)$ counts the number of paths of length $k$ with one flat step, and $p^{k}\left(\delta+\delta^j \right)$ counts the number of paths of length $k$ with one flat step at a level $y$ and another flat step at the level $y+j$.%DECIDE

\begin{rem}\label{shift_rmk}
Clearly, $p^{k}\left(\beta^{i}\right)=p^{k}\left(\beta\right)$,
for any multi-index $\beta$ and integer $i$.
\end{rem}

\begin{figure}
  \centering
  \begin{tikzpicture}
    \coordinate (XAxisMin) at (-1,0);
    \coordinate (XAxisMax) at (5,0);
    \coordinate (YAxisMin) at (0,-1);
    \coordinate (YAxisMax) at (0,3);
    \draw [thick, black] (-0.5,0) -- (1.5,0);% Draw x axis
    \draw [thick, black] (0,-0.5) -- (0,1.5);% Draw y axis

    \draw[style=help lines,dashed] (-1,-2) grid[step=1cm] (5,3); % Grid

    \draw [thick, black,-latex] (0,0) -- (1,1); % Path
    \draw [thick, black,-latex] (1,1) -- (2,1); % Path
    \draw [thick, black,-latex] (2,1) -- (3,2); % Path
    \draw [thick, black,-latex] (3,2) -- (4,1); % Path
    \draw [thick, black,-latex] (4,1) -- (5,0); % Path

    \node[draw,circle,inner sep=2pt,fill] at (0,0) {}; % Points
    \node[draw,circle,inner sep=2pt,fill] at (1,1) {};
    \node[draw,circle,inner sep=2pt,fill] at (2,1) {};
    \node[draw,circle,inner sep=2pt,fill] at (3,2) {};
    \node[draw,circle,inner sep=2pt,fill] at (4,1) {};
    \node[draw,circle,inner sep=2pt,fill] at (5,0) {};

  \end{tikzpicture}
  \caption{The path corresponding to $(0,1,1,2,1,0) \in \mathcal{P}^5\left( \delta \right)$}
\end{figure}

\begin{figure}
  \centering
  \begin{tikzpicture}
    \coordinate (XAxisMin) at (-1,0);
    \coordinate (XAxisMax) at (6,0);
    \coordinate (YAxisMin) at (0,-1);
    \coordinate (YAxisMax) at (0,3);
    \draw [thick, black] (-0.5,0) -- (1.5,0);% Draw x axis
    \draw [thick, black] (0,-0.5) -- (0,1.5);% Draw y axis

    \draw[style=help lines,dashed] (-1,-2) grid[step=1cm] (6,3);

    \draw[thick, black,-latex] (0,0) -- (1,1); % Path
    \draw[thick, black,-latex] (1,1) -- (2,0); % Path
    \draw[thick, black,-latex] (2,0) -- (3,-1); % Path
    \draw[thick, black,-latex] (3,-1) -- (4,-1); % Path
    \draw[thick, black,-latex] (4,-1) -- (5,0); % Path
    \draw[thick, black,-latex] (5,0) -- (6,0); % Path

    \node[draw,circle,inner sep=2pt,fill] at (0,0) {};
    \node[draw,circle,inner sep=2pt,fill] at (1,1) {};
    \node[draw,circle,inner sep=2pt,fill] at (2,0) {};
    \node[draw,circle,inner sep=2pt,fill] at (3,-1) {};
    \node[draw,circle,inner sep=2pt,fill] at (4,-1) {};
    \node[draw,circle,inner sep=2pt,fill] at (5,0) {};
    \node[draw,circle,inner sep=2pt,fill] at (6,0) {};
  \end{tikzpicture}
  \caption{The path corresponding to $(0,1,0,-1,-1,0,0) \in \mathcal{P}^6\left( \delta+\delta^1 \right)$}
\end{figure}

We show in Proposition \ref{coef_calc} below that for odd $k$
$$p^k\left( \delta \right)=\frac{k!}{\left(\frac{k-1}{2} \right)!^2},$$
and for even $k$ and $j \in \mathbb{N}$
$$p^{k}\left(\delta+\delta^{j}\right)=2\cdot\sum_{i,\ell,m}\binom{2m+i}{m}\binom{2\ell+j}{\ell}\binom{k-2-2m-i-2\ell-j}{\frac k2-1-m-\ell},$$
where the sum is extended over all integers $i,\ell,m$ for which the
expression is meaningful (i.e.\ $m,m+i,\ell,\ell+j,m+\ell,m+\ell+i+j\in\left[0,\frac{k}{2}-1\right]$).

We now break Theorem \ref{thm:main} into four propositions.

\begin{propos} \label{propos:NoCLT}
For any odd positive integer, $m$, any $\alpha>\alpha_c^\bullet$, and $P \in \bullet_m$ (where $\bullet=\mathcal{Q}, \mathcal{E}, \mathcal{Q}^\perp$),
$\Tr{P(H_{\alpha,N})}-\Ex{\Tr{P(H_{\alpha,N})}}$ converges a.s.\ to a random variable with finite variance.
\end{propos} % SHOULD THIS TAKE THE DAMNED COUNTEREXAMPLE INTO ACCOUNT???

For the next propositions, we will use the following notation: given an $(m+1)\times (m+1)$ matrix $\left( C(k,\ell)\right)_{k,\ell=0}^m$ and $P(x)=\sum_{k=0}^m a_k x^k \in\mathcal{V}_m$ we define
\begin{equation} \nonumber
CP(x)=\sum_{k=0}^m \sum_{\ell=0}^m C(k,\ell)a_\ell x^k.
\end{equation}
Recall that $Var(X_n)=\eta^2$.

\begin{propos}\label{propos:odd}
For any odd $k\in\N$ and $0< \alpha \le \frac12$,
$$\frac{\Tr{H_{\alpha,N}^k}-\Ex{\Tr{H_{\alpha,N}^k}}}{g_{2\alpha}(N)}\overset{d}{\longrightarrow} N\left(0,\sigma_k^2\right)$$
as $N \rightarrow \infty$, where
$$\sigma_k^2=p^k\left( \delta \right)^2\cdot \eta^2.$$%DECIDE (\eta is useless in the cumbersome 2.7, kind of odd in 2.5, and perhaps should be avoided altogether outside of the Intro, for the sake of consistency)

Moreover, for odd $k,\ell$,
\begin{equation} \label{eq:CQR}
\begin{split}
\underset{N\longrightarrow\infty}{\lim}&\Cov{\frac{\Tr{H_{\alpha,N}^k}-\Ex{\Tr{H_{\alpha,N}^k}}}{g_{2\alpha}(N)},\frac{\Tr{H_{\alpha,N}^\ell}-\Ex{\Tr{H_{\alpha,N}^\ell}}}{g_{2\alpha}(N)}}\\
&= p^k(\delta)p^\ell(\delta) \eta^2.
\end{split}
\end{equation}%APPROVE (equation is wider with nice 'E' and 'Tr' and 'Cov', so I dropped the C(k,l), we didn't use that notation anyway). Old version follows:
%Moreover, for $k,\ell\in \N$ that are odd, let
%\begin{equation} \label{eq:CQR}
%\begin{split}
%C(k,\ell)&=\underset{N\longrightarrow\infty}{\lim}Cov\left(\frac{Tr(H_{\alpha,N}^k)-\E[Tr(H_{\alpha,N}^k)]}{g_{2\alpha}(N)},\frac{Tr(H_{\alpha,N}^\ell)-\E[Tr(H_{\alpha,N}^\ell)]}{g_{2\alpha}(N)}\right)\\
%&=p^k(\delta)p^\ell(\delta) \eta^2.
%\end{split}
%\end{equation}

Letting $C_{\mathcal O}^{(m)}(k,\ell)$ be defined by \eqref{eq:CQR} in the case that $k$ and $\ell$ are both odd, and $0$ otherwise, it follows that for any $P \in \mathcal{O}_m$

\begin{equation} \nonumber
\frac{\Tr{P\left(H_{\alpha,N}\right)}-\Ex{\Tr{P\left(H_{\alpha,N}\right)}}}{g_{2\alpha}(N)} \overset{d}{\longrightarrow} N\left(0,\sigma_{\mathcal{O}}(P)^2 \right)
\end{equation}
as $N \rightarrow \infty$, where
\begin{equation} \label{eq:varianceQ}
\sigma_{\mathcal{O}}(P)^2=\left \langle P, C_{\mathcal{O}}^{(m)} P \right \rangle=\left|\left \langle P,Q_m \right \rangle \right|^2 \eta^2.
\end{equation}

\end{propos}

\begin{propos}\label{propos:even}
For any even $k \in \N$ and $0< \alpha \le \frac14$,
$$\frac{\Tr{H_{\alpha,N}^k}-\Ex{\Tr{H_{\alpha,N}^k}}}{g_{4\alpha}(N)}\overset{d}{\longrightarrow} N\left(0,\sigma_k^2\right)$$
as $N \rightarrow \infty$, where
$$\sigma_k^2=\left(p^k(2\delta)\right))^2 \cdot \left(\Ex{X_1^4}-\eta^4 \right)+\left(\sum_{s=1}^{\infty}\left(p^k(\delta+\delta^s)\right)^2\right)\cdot \eta^4.$$
If $k\ge 4$, then $\sigma_k^2>0$ (if $\Var{X_n^2}>0$ then we can allow $k=2$ as well).

Moreover, for even $k,\ell$,
\begin{equation} \label{eq:CE}
\begin{split}
\underset{N\longrightarrow\infty}{\lim}&\Cov{\frac{\Tr{H_{\alpha,N}^k}-\Ex{\Tr{H_{\alpha,N}^k}}}{g_{4\alpha}(N)},\frac{\Tr{H_{\alpha,N}^\ell}-\Ex{\Tr{H_{\alpha,N}^\ell}}}{g_{4\alpha}(N)}}\\
&=p^k(2\delta)p^\ell(2\delta) \left(\Ex{X_1^4}-\eta^4 \right)+\left(\sum_{s=1}^\infty p^k(\delta+\delta^s) p^\ell(\delta+\delta^s)\right)\eta^4.
\end{split}
\end{equation}%APPROVE (same change as in propos:odd). Old version follows:
%Moreover, for even $k,\ell$, let
%\begin{equation} \label{eq:CE}
%\begin{split}
%C(k,\ell)&=\underset{N\longrightarrow\infty}{\lim}Cov\left(\frac{Tr(H_{\alpha,N}^k)-\E[Tr(H_{\alpha,N}^k)]}{g_{4\alpha}(N)},\frac{Tr(H_{\alpha,N}^\ell)-\E[Tr(H_{\alpha,N}^\ell)]}{g_{4\alpha}(N)}\right)\\
%&=\left(p^k(2\delta)p^\ell(2\delta) \left(\E[X_1^4]-\eta^4 \right)+\sum_{i=1}^\infty p^k(\delta+\delta^i) p^\ell(\delta+\delta^i)\eta^2\right).
%\end{split}
%\end{equation}

Letting $C_{\mathcal{E}}^{(m)}(k,\ell)$ be defined by \eqref{eq:CE} in the case that $k$ and $\ell$ are both even, and $0$ otherwise, it follows that for any $P \in \mathcal{E}_m$,

\begin{equation} \nonumber
\frac{\Tr{P\left(H_{\alpha,N}\right)}-\Ex{\Tr{P\left(H_{\alpha,N}\right)}}}{g_{4\alpha}(N)} \overset{d}{\longrightarrow}N\left(0,\sigma_{\mathcal{E}}(P)^2 \right)
\end{equation}
as $N \rightarrow \infty$, where
\begin{equation} \label{eq:varianceE}
\sigma_{\mathcal{E}}(P)^2=\left \langle P, C_{\mathcal{E}}^{(m)} P \right \rangle.
\end{equation}
If $\deg(P)\ge 4$, or $\Var{X_n^2}>0$ and $P$ is not constant, then $\sigma_{\mathcal{E}}(P)^2>0$.
\end{propos}

For an odd $k \geq 3$, define
\begin{equation} \label{perpbasis}
P_k(x)=x^k-p^k(\delta)x.
\end{equation}

\begin{lemma} \label{lemma:perpdegree}
For any odd $m$, $P_3,P_5,...,P_m$ are a basis for $\mathcal{Q}_m^{\perp}$.
\end{lemma}

\begin{proof}
Since $P_3,...,P_m$ are linearly independent and they span a subspace of dimension $\frac{m-1}{2}$, it is enough to show that these polynomials are in $\mathcal{Q}_m^{\perp}$.
Observe that $Q_m(x)=\sum_{j=0}^m p^j(\delta)x^j$ and $p^1(\delta)=1$, therefore
$$\left<P_k,Q_m \right>=p^k(\delta)\cdot 1 - 1\cdot p^k(\delta)=0.$$

\end{proof}

The fluctuations of $P \in \mathcal{Q}_m^{\perp}$ involve several types of paths. We use the following notation to categorize and enumerate them, where $k,\ell\geq 3$ are odd integers.
\begin{equation} \label{eq:q_6}
q_6=p^k(3\delta)p^{\ell}(3\delta).
\end{equation}
\begin{equation} \label{eq:q_32}
q_{3,3}=\sum_{t=1}^\infty \left(p^k(2\delta+\delta^t)p^\ell(\delta+2\delta^t) + p^\ell(2\delta+\delta^t)p^k(\delta+2\delta^t)\right).
\end{equation}
\begin{equation} \label{eq:q_42}
\begin{split}
q_{4,2}& =\sum_{t=1}^\infty \left(p^k(3\delta)p^\ell(2\delta+\delta^t) + p^\ell(3\delta)p^k(2\delta+\delta^t) \right) \\
&\quad + \sum_{s=1}^\infty \left(p^k(3\delta)p^\ell(\delta+2\delta^s) + p^\ell(3\delta)p^k(\delta+2\delta^s)\right) \\
&\quad + \sum_{t=1}^\infty p^k(2\delta +\delta^t)p^\ell(2\delta +\delta^t)\\
&\quad + \sum_{s=1}^\infty p^k(\delta+2\delta^s)p^\ell(\delta+2\delta^s). \\
\end{split}
\end{equation}
\begin{equation} \label{q_222}
\begin{split}
q_{2,2,2}&=\sum_{s,t=1}^\infty \left(p^k(2\delta+\delta^{s+t})p^\ell(2\delta+\delta^t) + p^\ell(2\delta+\delta^{s+t})p^k(2\delta+\delta^t) \right)\\
&\quad + \sum_{s,t=1}^\infty \left(p^k(2\delta+\delta^t)p^\ell(\delta+2\delta^s) + p^\ell(2\delta+\delta^t)p^k(\delta+2\delta^s)\right) \\
&\quad + \sum_{s,t=1}^\infty \left(p^k(\delta+2\delta^s)p^\ell(\delta+2\delta^{s+t}) + p^\ell(\delta+2\delta^s)p^k(\delta+2\delta^{s+t})\right) \\
&\quad + \sum_{s,t=1}^\infty p^k(\delta+\delta^s+\delta^{s+t})p^\ell(\delta+\delta^s+\delta^{s+t}).
\end{split}
\end{equation}
Note that $p^k\left( \delta+\delta^s+\delta^{s+t}\right)=0$ whenever $s+t>(k-3)/2$, so all the sums above are in fact finite.

\begin{propos}\label{propos:perp}
Let $0<\alpha\leq\frac16$. For any odd $k,\ell\geq 3$, let
\begin{equation} \label{eq:M_kl}
\begin{split}
M_{k,\ell}& = q_6\left(\Ex{X_1^6}-\Ex{X_1^3}^2\right)+q_{3,3}\Ex{X_1^3}^2 \\
&\quad  + q_{4,2}\Ex{X_1^4}\Ex{X_1^2}+q_{2,2,2}\Ex{X_1^2}^3.
\end{split}
\end{equation}

Then for any odd $k\geq 3$,
$$\frac{\Tr{P_k(H_{\alpha,N})}-\Ex{\Tr{P_k(H_{\alpha,N})}}}{g_{6\alpha}(N)}\overset{d}{\longrightarrow} N\left(0,\sigma_{\mathcal Q^\perp}(P_k)^2 \right),$$
where $\sigma_{\mathcal Q^\perp}(P_k)^2=M_{k,k}>0$.

Moreover, for odd $k,\ell$
\begin{equation} \nonumber
\begin{split}
& \underset{N\longrightarrow\infty}{\lim}\Cov{\text{\scalebox{0.9}
{$\frac{\Tr{P_k(H_{\alpha,N})}-\Ex{\Tr{P_k(H_{\alpha,N})}}}{g_{6\alpha}(N)},\frac{\Tr{P_\ell(H_{\alpha,N})}-\Ex{\Tr{P_\ell(H_{\alpha,N})}}}{g_{6\alpha}(N)}$}}}\\
& \quad =M_{k,\ell}.
\end{split}
\end{equation}

Fixing an odd $m$, for any $0\leq k,\ell\leq m$, letting $C_{\mathcal Q^\perp}^{(m)}(k,\ell)=M_{k,\ell}$ if $k$ and $\ell$ are both odd and $\geq 3$, and $C_{\mathcal Q^\perp}^{(m)}(k,\ell)=0$ otherwise, it follows that for any $P \in \mathcal{Q}_m^{\perp}$,

\begin{equation} \nonumber
\frac{\Tr{P\left(H_{\alpha,N}\right)}-\Ex{\Tr{P\left(H_{\alpha,N}\right)}}}{g_{6\alpha}(N)} \overset{d}{\longrightarrow} N\left(0,\sigma_{\mathcal Q^\perp}(P)^2 \right)
\end{equation}
as $N \rightarrow \infty$, where
\begin{equation} \label{eq:varianceE}
\sigma_{\mathcal{Q}^{\perp}}(P)^2=\left \langle P, C_{\mathcal Q^{\perp}}^{(m)} P \right \rangle .
\end{equation}
If $\Var{X_n^2}=0$ and $P(x)=a\left(\frac{x^5}{5}-4x^3+18x\right)$ for some constant $a$, then $\sigma_{\mathcal Q^\perp}(P)^2=0$. Otherwise, $\sigma_{\mathcal Q^\perp}(P)^2>0$.
%If $\Var{X_n^2}=0$ and $P(x)=a\left(\frac{x^5}{5}-4x^3+18x\right)$ for some $a$, then $\sigma_{\mathcal Q_m^\perp}(P)^2=0$, and $\sigma_{\mathcal Q_m^\perp}(P)^2>0$ otherwise.\\
%We have $\sigma_{\mathcal{Q}_m^\perp}(P)^2>0$ unless $\Var{X_n^2}=0$ and $P(x)=a\left(\frac{x^5}{5}-4x^3+18x\right)$ for some $a$. In the case $\Var{X_n^2}=0$ and $P(x)=a\left(\frac{x^5}{5}-4x^3+18x\right)$, we have $\sigma_{\mathcal{Q}_m^{\perp}}(P)^2=0$.\\
%We have $\sigma_{\mathcal{Q}_m^\perp}(P)^2=0$ if $\Var{X_n^2}=0$ and $P(x)=a\left(\frac{x^5}{5}-4x^3+18x\right)$ for some $a$, and $\sigma_{\mathcal Q_m^\perp}(P)^2>0$ otherwise.

\end{propos}
%\begin{rem}
%Note that all sums appearing in Proposition \ref{propos:perp} are in fact finite sums, since for a fixed $k$, there are only a finite number of multi-indices $\beta$ such that $\iota(\beta)=0$ and $p^k(\beta)\neq 0.$
%\end{rem}

Assuming the propositions above, we can now present the

\begin{proof}[Proof of Theorem \ref{thm:main}]
Part ii) of all three cases is proved by Proposition \ref{propos:NoCLT}.
Part i) of b) and c) are proved by Propositions \ref{propos:even} and \ref{propos:perp}. To prove a), use $\mathcal{V}_m=\mathcal{Q}_m\oplus \mathcal{Q}_m^\perp\oplus \mathcal{E}_m$ to decompose $P=P_{\mathcal Q_m}+P_{\mathcal Q_m^\perp}+P_{\mathcal E_m}$.
From Propositions \ref{propos:even} and \ref{propos:perp}, we see that
$$\frac{\Tr{P_{\mathcal E_m}(H_{\alpha,N})}+\Tr{P_{\mathcal Q_m^\perp}(H_{\alpha,N})} - \Ex{\Tr{P_{\mathcal E_m}(H_{\alpha,N})}+\Tr{P_{\mathcal Q_m^\perp}(H_{\alpha,N})}}}{g_{2\alpha}(N)}$$
vanishes in the limit. Part i) of a) thus follows from Proposition \ref{propos:odd} applied to $P_{\mathcal Q_m}$.
\end{proof}
%%%%%%%%%%%%%%%%%%%%%%%%%%%%%%%%%%%%%%%%%%%%%%%%%%%%%%%%%%% Section 3 %%%%%%%%%%%%%%%%%%%%%%%%%%%%%%%%%%%%%%%%%%%%%%%%

\section{Coefficient calculations}

In order to associate the paths described in the previous section with terms of the form $V^\beta$ in $\Tr{ H_{\alpha,N}^k}$, we first write $H_{\alpha,N}=S_N+V_{\alpha,N}+S^*_N$ where $S$ is the right-shift operator. Note that
\begin{equation} \label{eq:expansion}
\begin{split}
H_{\alpha,N}^k & =(S_N+V_{\alpha,N}+S^*_N)^k \\
 & =S_N^k+S_N^{k-1}V_{\alpha,N}+S_N^{k-1}S_N^*+S_N^{k-2}V_{\alpha,N} S_N^*+\ldots
\end{split}
\end{equation}
By replacing $S_N$ by $U$ and $S_N^*$ by $D$ in the summands above we get a sum of strings of $U$'s, $V_{\alpha,N}$'s and $D$'s. Since $N$ and $\alpha$ will be fixed in this section, we henceforth omit them from the notation for simplicity. Thus,
\begin{equation}
\eqref{eq:expansion}=U^k+U^{k-1}V+U^{k-1}D+U^{k-2}VD+\ldots
\end{equation}
By associating $U$ with an up step, $D$ with a down step, and $V$ with a flat step, we get a bijection between the following three objects:
\begin{itemize}
\item Matrices $M$ which appear in the expansion above
\item Strings of length $k$, consisting of the letters $U,V,D$,
\item Paths $\left(0=y_{0},y_{1},\ldots,y_{k}\right)$ starting at $0$ (so that, e.g., the string $UVUDD$ corresponds to the path in figure 1 above).
\end{itemize}

The contributions of each individual matrix in the expansion to the
total sum is outlined in the following lemma, which is easily verified:

\begin{lem}
Assume $k<N$. Then:
\begin{enumerate} \label{path_lem}
\item A path $\left(0=y_{0},y_{1},\ldots,y_{k}\right)$ has $y_{k}=m$,
if and only if there is a non-negative integer, $\ell \geq -m$, such that in the corresponding string of length $k$, the letter
$D$ appears $\ell$ times, and $U$ appears $\ell+m$ times. Equivalently, the corresponding
matrix $M$ has $M_{ij}=0$ whenever $j\neq i+m$.
\item A matrix $M$ has $M_{ij}\not \equiv 0$ (as a function of $V$) if and only if $y_{k}=j-i$ in the
corresponding path, and additionally $1\leq i+y_{t}\leq N$ for every
$1\leq t\leq k$.

Furthermore, under these conditions, $M_{ij}=V^{\beta}$, where for every integer $h$, the
number of flat steps at level $h$ is $\beta_{h+i}$ $($in the case that $\beta=0$, we interpret $V^\beta=1)$. In other words,
$M_{ij}=V^{\beta}$ if and only if the multi-index $\beta^{-i}$ counts
the number of flat steps (at each level) in the corresponding path.
\end{enumerate}
\end{lem}

\begin{cor} \label{oddeven}
If $j=i+m$ and $k-m$ is even (odd), the polynomial $\left(H_{\alpha,N}^k\right)_{\,i,j}$
is a finite sum of monomials of even (odd) degree.
\end{cor}
\begin{proof}
By part 1 of Lemma \ref{path_lem}, non-zero contributions to $\left(H_{\alpha,N}^k\right)_{\,i,j}$
come from matrices $M$ in the expansion, for which in the corresponding
string some $\ell$ entries are $D$, another $\ell+m$ entries are $U$,
thus the remaining $k-2\ell-m$ entries are $V$. By part 2 of Lemma \ref{path_lem}, $M_{ij}$
is a monomial whose degree equals the total number of flat steps in the path, which is $k-2\ell-m$.
\end{proof}
\begin{cor} \label{oddeven2}
If $k$ is odd (even), then for any $1 \le j \le N$, $\left(H_{\alpha,N}^k\right)_{\,j,j}$ is a finite sum of monomials of an odd (even) degree.
\end{cor}

We now turn our attention to diagonal entries of $H_{\alpha,N}^k$.
For corresponding strings and paths, this means we restrict our attention
to strings of length $k$ with an equal number of $D$'s and $U$'s,
and paths of length $k$ ending at $y_{k}=0$.

Note that if $\frac{k}{2}\leq i\leq N-\frac{k}{2}$, then
for every $1\leq t\leq k$ and every relevant path, we have $1\leq i+ y_{t}\leq N$.
We conclude:

\begin{cor}
If $\frac{k}{2}\leq i\leq N-\frac{k}{2}$, the coefficient of $V^{\beta}$
in the polynomial $\left(H_{\alpha,N}^k\right)_{i,i}$ equals
the number of paths $\left(y_{0},y_{1},\ldots,y_{k}\right)$ with
$y_{0}=y_{k}=0$, that have precisely $\beta_{h+i}$ flat steps at
each level $h$.
\end{cor}
\begin{proof}
Use part 2 of Lemma \ref{path_lem} and sum $M_{ii}$ from every matrix $M$ appearing
in the expansion of $H_{\alpha,N}^k$.
\end{proof}

We recall that for a multi-index $\beta$, $p^{k}\left(\beta\right)$ denotes
the number of paths $\left(0=y_{0},y_{1},\ldots,y_{k}=0\right)$,
such that $\beta^{i}$ counts flat steps in the path, for some integer
$i$. For a non-zero multi-index, $\beta$, let $\iota(\beta)$ denote the minimal $i\in\mathbb{Z}$ for which $\beta_i>0$.

\begin{propos} \label{coef_prop}
If $\beta\neq0$ has $\iota(\beta)=\iota\in[k,N-k]$, then the coefficient of $V^{\beta}$ in the polynomial $\Tr{H_{\alpha,N}^{k}}$ is $p^{k}\left(\beta\right)$.
\end{propos}
\begin{proof}
We claim that any path $\left(0=y_{0},y_{1},\ldots,y_{k}=0\right)$ for which $\beta^{-i}$ counts the number of flat steps, is uniquely matched with a matrix $M$ from the expansion of $H_{\alpha,N}^{k}$ satisfying $M_{ii}=V^{\beta}$.
Indeed, a path $\left(y_{0},y_{1},\ldots,y_{k}\right)$ only has flat steps from the range $\left(-\frac{k}{2},\frac{k}{2}\right)$. We know that $\beta^{-i}_{\iota-i}=\beta_\iota$ is non-zero, therefore $-\frac{k}{2}\leq\iota-i\leq\frac{k}{2}$. From $k\leq\iota\leq N-k$, we deduce $\frac{k}{2}\leq i\leq N-\frac{k}{2}$, which in turn guarantees that $1\leq y_{t}+i\leq N$ for every $t=0,1,2\ldots,k$.
\end{proof}
Finally, we count paths to obtain explicit formulas for the coefficients in the polynomial $\Tr{H_{\alpha,N}^{k}}$.

\begin{propos}\label{coef_calc}\
\begin{enumerate}
\item If $k$ is even,
\[
p^{k}\left(0\right)=\binom{k}{\frac{k}{2}}
\]
(here $0$ denotes the zero function, as a multi-index).
\item If $k$ is odd,
\begin{equation} \label{odd_delta}
p^{k}\left(\delta\right)=k\cdot\binom{k-1}{\frac{k-1}{2}}=\frac{k!}{\left(\frac{k-1}{2}\right)!^{2}}
\end{equation}
\item If $k$ is even,
\[
\sum_{j\geq0}p^{k}\left(\delta+\delta^{j}\right)=\binom{k}{2}\cdot\binom{k-2}{\frac{k-2}{2}}
\]
\item If $k$ is even,
\[
p^{k}\left(2\delta\right)=k\cdot2^{k-3}
\]
\item If $k$ is even and $j\in\mathbb{N}$,
\[
p^{k}\left(\delta+\delta^{j}\right)=2\cdot\sum_{i,\ell,m}\binom{2m+i}{m}\binom{2\ell+j}{\ell}\binom{k-2-2m-i-2\ell-j}{\frac k2-1-m-\ell}
\]
where the sum is extended over all integers $i,\ell,m$ for which the
expression is meaningful (i.e.\ $m,m+i,\ell,\ell+j,m+\ell,m+\ell+i+j\in\left[0,\frac{k}{2}-1\right]$).
\end{enumerate}
\end{propos}
\begin{proof}\
\begin{enumerate}
\item Paths $\left(y_{0},y_{1},\ldots,y_{k}\right)$ counted by $p^{k}\left(0\right)$
correspond to strings of length $k$ with the letters $D,U$, where
$D$ and $U$ appear an equal number of times.

\item Paths $\left(y_{0},y_{1},\ldots,y_{k}\right)$ counted by $p^{k}\left(\delta\right)$
correspond to strings of length $k$ with the letters $D,V$, and $U$,
where $V$ appears exactly once, and $D,U$ appear an equal
number of times. We choose the position for $V$ from $k$ options,
then divide the remaining $k-1$ positions into two parts of size
$\frac{k-1}{2}$ to place $U$'s and $D$'s.

\item Paths $\left(y_{0},y_{1},\ldots,y_{k}\right)$ counted in this sum
are paths of length $k$ with exactly $2$ flat steps, corresponding
to strings of length $k$ where $V$ appears exactly twice,
and $U,D$ appear an equal number of times. (If the two steps are
at level $i$ and $i+j$, the path is counted by $p^{k}\left(\delta+\delta^{j}\right)$).
We choose two positions for $V$ and divide the remaining $k-2$
positions evenly among $U$ and $D$.

\item $p^{k}\left(2\delta\right)$ counts paths of length $k$, which have
exactly $2$ flat steps, at the same level. Such paths correspond
to composite strings of the form $s_{1}V s_{2}V s_{3}$,
where $s_{2}$ and $s_{1}s_{3}$ are strings composed from the letters
$U,D$, each appearing an equal number of times. There are ${\displaystyle \left(2m+1\right)\binom{2m}{m}\binom{k-2m-2}{\frac{k}{2}-m-1}}$
strings of this form in which $s_{1}s_{3}$ has length $2m$, therefore
\[
p^{k}(2\delta)=\sum_{m=0}^{\frac{k}{2}-1}\left(2m+1\right)\binom{2m}{m}\binom{k-2m-2}{\frac{k}{2}-m-1}.
\]
Using standard generating function techniques, one can verify that
\[
f\left(t\right)=\sum_{m\geq0}\binom{2m}{m}t^{m}=\left(1-4t\right)^{-\frac{1}{2}}
\]
and deduce that
\begin{align*}
\sum_{m\geq0}\binom{2m}{m}\left(2m+1\right)t^{m} & =2t\sum_{m\geq0}\binom{2m}{m}m\cdot t^{m-1}+\sum_{m\geq0}\binom{2m}{m}t^{m}\\
 & =2t\cdot f'\left(t\right)+f\left(t\right)\\
 & =\left(1-4t\right)^{-\frac{3}{2}}
\end{align*}
thus
\begin{align*}
\sum_{n\geq0}\left(\sum_{m\geq0}\left(2m+1\right)\binom{2m}{m}\binom{2\left(n-m\right)}{n-m}\right)t^{n} & =\left(1-4t\right)^{-\frac{1}{2}}\cdot\left(1-4t\right)^{-\frac{3}{2}}\\
 & =\frac{1}{\left(1-4t\right)^{2}}\\
 & =\frac{1}{4}\left(\frac{1}{1-4t}\right)' \\
 & =\sum_{n\geq0}\left(n+1\right)\cdot4^{n}t^{n}
\end{align*}
and our proof is completed by equating the coefficient of $t^{n}$
for $n={\displaystyle \frac{k}{2}}-1$.

\item $p^k(\delta+\delta^j)$ counts paths of length $k$, with exactly $2$ flat
steps: one at level $i$ and another at level $i+j$. If the first
step is at level $i$, we are counting strings of the form $s_{1}V s_{2}V s_{3}$,
where $s_{1},s_{2},s_{3}$ consist of the letters $U,D$, and:
\begin{itemize}
\item In $s_{1}$, the letter $D$ appears $m$ times and $U$ appears $m+i$
times,
\item In $s_{2}$, the letter $D$ appears $\ell$ times and $U$ appears $\ell+j$
times,
\item In $s_{3}$, the letter $D$ appears $\frac k2-1-m-\ell$ times, and
$U$ appears $\frac k2-1-m-i-\ell-j$ times.
\end{itemize}
This yields the sum $${\displaystyle \sum_{i,m,\ell}\binom{2m+i}{m}\binom{2\ell+j}{l}\binom{k-2-2m-i-2\ell-j}{\frac k2-1-m-\ell}}.$$
The factor $2$ comes from paths where the step at level $i+j$ appears
before the step at level $i$.
\end{enumerate}
\end{proof}

%%%%%%%%%%%%%%%%%%%%%%%%%%%%%%%%%%%%%%%%%%%%%%%%%%%%%%%%%%%%%%%%%%%%%%%%%%%%%%%%%% Section 4 %%%%%%%%%%%%%%%%%%%%%%%%%%

\section{Fluctuation Asymptotics}
In this section we will prove Propositions \ref{propos:NoCLT}, \ref{propos:odd}, \ref{propos:even}, and \ref{propos:perp}. We start with the simpler Proposition \ref{propos:NoCLT}.

\subsection{Proof of Proposition \ref{propos:NoCLT}}

The key to proving Proposition \ref{propos:NoCLT} is the following lemma \cite[Theorem 2.5.3]{Durett}:
\begin{lemma} \label{thm:kol}
Suppose $Y_1,Y_2,...$ are independent random variables with $\Ex{Y_n}=0$. If 
$$\sum_{n=1}^\infty \Var{Y_n}<\infty, $$
then with probability $1$, $\sum_{n=1}^\infty Y_n$ converges.
\end{lemma}%DECIDE: Consider using Y_i instead, for consistent notation

We also need the definition of $m$-dependent random variables.

\begin{defin}
Let $Y_1,Y_2,...$ be a sequence of random variables. We say that the sequence is $m$-dependent, if the inequality $s-r>m$ implies that the two sets
$$(Y_1,Y_2,...,Y_r),\quad \quad (Y_s,Y_{s+1},...,Y_n) $$
are independent for any $n\ge s$.

\end{defin}%DECIDE: Consider using Y_i instead, for consistent notation

\begin{lemma}\label{using:kol}
Let $k\in \N$ and let $\beta$ be a multi-index with $|\beta|=n$, supported on $[0,k-1]$. Fix $\alpha>\frac{1}{2n}$. Then

$$\sum_{i=1}^N  V_\alpha^{\beta^i}-\E\left[V_\alpha^{\beta^i} \right]$$
converges a.s.\ to a random variable $Q$ as $N\rightarrow \infty $.

\end{lemma}

\begin{proof}
Define $Y_i(\beta)=V_\alpha^{\beta^i}-\Ex{V_\alpha^{\beta^i}}$. Since $\beta^i$ is supported on $[i,i+k-1]$, the sequence $Y_1(\beta),Y_2(\beta),...$ is $k-1$ dependent. For every $1 \le \ell \le k$,%Changed support range
$$\sum_{i=0}^{\infty} \Var{Y_{{i\cdot k+\ell}}(\beta)} =\sum_{i=0}^{\infty} \Var{V_\alpha^{\beta^{i\cdot k+\ell}}}\le \sum_{i=1}^{\infty}\frac{C}{i^{2n\alpha}}<\infty $$
for some constant $C>0$. Since the conditions in Lemma \ref{thm:kol} hold ($\Ex{Y_i(\beta)}=0$), for every $1\le \ell \le k$ each of the series $\sum_{i=1}^{\infty} Y_{i\cdot k+\ell}(\beta)$ converges a.s.\ to a random variable, and therefore $\sum_{i=1}^{\infty} Y_i(\beta) $ converges a.s.
\end{proof}

Recall that $\iota(\beta)$ is the minimal $i\in \Z$ for which $\beta_i>0$. Also recall that $p^k(\beta)$ counts the number of length $k$ paths, starting and ending at $0$, whose flat steps have a structure described by $\beta$. Clearly, for any fixed $k$, there is a finite number of paths of length $k$ starting and ending at $0$, and so in particular, various sums of the form $\sum_{\iota(\beta)=0} p^k(\beta)$ appearing below, are finite.

\begin{lemma}\label{ARV}
For any $k\in\N$, define
\begin{equation*}
\begin{split}
A^k(N)\equiv \Tr{H_{\alpha,N}^{k}}&-\sum_{\iota(\beta)\in[1,N]}p^{k}\left(\beta\right)V_\alpha^{\beta}\\
&-\Ex{\Tr{H_{\alpha,N}^{k}}-\sum_{\iota(\beta)\in[1,N]}p^{k}\left(\beta\right)V_\alpha^{\beta}}.
\end{split}
\end{equation*}
Then $A^k(N)$ converges almost surely (when $N\rightarrow\infty$) to a bounded random variable.
\end{lemma}

\begin{proof}
From Proposition \ref{coef_prop}, we can decompose $A^k(N)=A^k_1(N)+A^k_2(N)$, where $A^k_1(N)$ is a polynomial in the variables $V_\alpha(1),\ldots,V_\alpha(2k)$ (this polynomial is fixed for $N>k$), and $A^k_2(N)$ is a polynomial in the variables $V_\alpha(N-2k),\ldots,V_{\alpha}(N)$. Therefore $A^k_2(N)\longrightarrow 0$ as $N\longrightarrow \infty$.
\end{proof}

\begin{proof}[Proof of Proposition \ref{propos:NoCLT}]
First use Remark \ref{shift_rmk} to write
\begin{equation} \label{eq:Sum2.3}
\begin{split}
&\Tr{H_{\alpha,N}^{k}}-\Ex{\Tr{H_{\alpha,N}^{k}}} \\
& =\sum_{\iota(\beta)\in[1,N]}p^{k}\left(\beta\right)V_\alpha^{\beta}-\Ex{\sum_{\iota(\beta)\in[1,N]}p^{k}\left(\beta\right)V_\alpha^{\beta}}+A^k(N)\\
& =\sum_{\iota(\beta)=0}\sum_{i=1}^{N} p^{k}\left(\beta\right)V_\alpha^{\beta^{i}}-\Ex{\sum_{\iota(\beta)=0}\sum_{i=1}^{N} p^{k}\left(\beta\right)V_\alpha^{\beta^{i}}}+A^k(N)\\
& =\sum_{\iota(\beta)=0} p^k(\beta) \left(\sum_{i=1}^{N} \left(V_\alpha^{\beta^i}-\Ex{V_\alpha^{\beta^i}}\right) \right)+A^k(N).\\
\end{split}
\end{equation}

Now, in the cases $P\in\mathcal O_m$ and $P \in \mathcal{E}_m$ it clearly suffices to prove the statement for monomials. Considering $P(x)=x^k$ with $k$ odd, it follows from Lemmas \ref{using:kol} and \ref{ARV}, that \eqref{eq:Sum2.3} converges a.s.\ for any $\alpha>\frac{1}{2}$ since $\alpha>\frac12\geq\frac{1}{2|\beta|}$ holds for any $\beta\not\equiv 0$. This proves the statement for $P \in \mathcal{O}_m$.

Taking $P(x)=x^k$ for $k$ even we see that every monomial $V_\alpha^{\beta^i}$ appearing in \eqref{eq:Sum2.3} is of degree $2$ or more ($p^k(\beta)=0$ when $k$ is even and $\left|\beta\right|$ is odd). We can therefore assume $|\beta|\ge 2 $, and the conditions of Lemma \ref{using:kol} hold for $\alpha>\frac14$. This takes care of $P \in \mathcal{E}_m$.

Finally, for $P\in\mathcal Q_m^\perp$, we consider $P=P_k$ (where $P_k$ was defined in \eqref{perpbasis}) and obtain
\begin{equation*}
\begin{split}
&\Tr{P_k\left(H_{\alpha,N}\right)}-\Ex{\Tr{P_k\left(H_{\alpha,N}\right)}} \\
& =\sum_{\iota(\beta)=0} \left(p^k(\beta)-p^k(\delta)p^1(\beta)\right) \left(\sum_{i=1}^{N} \left(V_\alpha^{\beta^i}-\Ex{V_\alpha^{\beta^i}}\right) \right)\\
& +A^k(N)-p^k(\delta)A^1(N).
\end{split}
\end{equation*}
Note that monomials of degree less than $3$ vanish from the sum so we may assume $|\beta|\geq 3$ and the conditions of Lemma \ref{using:kol} hold for $\alpha>\frac16\geq\frac{1}{2|\beta|}$. Since the $P_k$'s are basis elements for $\mathcal{Q}_m^\perp$, this finishes the proof.
\end{proof}

\vspace{\baselineskip}
\subsection{Preliminaries for Limit Theorems}

For our limit theorems, we shall use the following result of Orey \cite{SO}.

\begin{propos} \label{thm:orey}
Let $Y_1,Y_2,..$ be a sequence of $m$-dependent random variables with $\Ex{Y_k}=0$ and $\Var{Y_k}<\infty$ for every $k \in \N$. Denote
$$D_n^2=\Var{Y_1+...+Y_n}, $$
and assume that
\begin{enumerate}
  \item $\frac{1}{D_n^2}\sum_{k=1}^n \E[Y_k^2\cdot 1_{\{|Y_k|>\epsilon D_n\}}] \longrightarrow 0$ for every $\epsilon>0$
  \item $\frac{1}{D_n^2}\sum_{k=1}^n \Var{Y_k}=O(1)$.
\end{enumerate}
Then
$$\frac{Y_1+...+Y_n}{D_n}\overset{d}{\longrightarrow}N(0,1) .$$
\end{propos}%DECIDE: Consider using Y_i or V_k instead, for consistent notation

Our  limit theorems follow from the next two lemmas. The first lemma allows us to compute the variance of the fluctuations in terms of path counting numbers and the moments of $X_n$, and the second lemma shows that these fluctuations converge to random variables with normal distribution. Proving that the variance is positive (and describing under which conditions it is positive) requires separate arguments for each different case (if the polynomial $P$ is in $\mathcal Q_m, \mathcal E_m$, or $\mathcal Q_m^\perp$).

Note that if a multi-index $\beta$ is supported on $\N$, we have (since $X_n$ are i.i.d.):
\begin{equation} \label{index_ex}
\Ex{X^\beta}=\prod_{n=1}^{\infty}\Ex{X_n^{\beta_n}}=\prod_{n=-\infty}^{\infty}\Ex{X_1^{\beta_n}},
\end{equation}
which is of course a finite product. It is convenient to use (\ref{index_ex}) as a definition of $\Ex{X^\beta}$ for any $\beta$ not necessarily supported on $\N$. Accordingly, denote $\Cov{X^\beta,X^\gamma}=\Ex{X^{\beta+\gamma}}-\Ex{X^\beta}\Ex{X^\gamma}$. Finally, for a multi-index $\beta$ and $i \in \N$, denote
$$i^\beta=\prod_{j=-\infty}^{\infty} (i+j)^{\beta_j}. $$
Using this notation, we prove

\begin{lemma} \label{norm_cov}
Let $\alpha>0$, and let $\beta,\gamma$ be multi-indices with $\iota(\beta)=\iota(\gamma)=0$ such that $2|\beta|\alpha,2|\gamma|\alpha\le 1$. Then there exists a constant $C(\alpha,\beta,\gamma)>0$ such that
\begin{equation} \label{cov_value}
\underset{N\longrightarrow \infty}{\lim} \Cov{\frac{\sum_{i=1}^{N} V_\alpha^{\beta^i}}{g_{2|\beta|\alpha}(N)} , \frac{\sum_{i=1}^{N} V_\alpha^{\gamma^i}}{g_{2|\gamma|\alpha}(N)}} =C(\alpha,\beta,\gamma)\cdot \sum_{j=-\infty}^{\infty} \Cov{X^\beta,X^{\gamma^j}}.
\end{equation}
If $|\beta|=|\gamma|$, then $C(\alpha,\beta,\gamma)=1$.

Note that $\Ex{X^{\beta+\gamma^j}} = \Ex{X^\beta}\Ex{X^{\gamma^j}}$ if $\beta$ and $\gamma^j$ have disjoint supports, so the sum in (\ref{cov_value}) is in fact finite.
\end{lemma}

\begin{proof}
For a fixed $j\in \Z$, define
$$I_N = \{i \ : \ 1\le i \le N, \ 1 \le i+j \le N \},$$
$$c_j(N)=\left\{ \begin{array}{cc} 0 & I_N=\emptyset \\ \sum_{i\in I_N} \frac{1}{(i^{\beta}(i+j)^\gamma)^{\alpha}} & \textrm{ otherwise.} \end{array} \right. $$
Then
$$\Cov{\sum_{i=1}^N V_\alpha^{\beta^i,},\sum_{i=1}^N V_\alpha^{\gamma^{i}}}=\sum_{j \in \Z} c_j(N)\Cov{X^{\beta^i},X^{\gamma^{(i+j)}}} $$
For multi-indices $\beta,\gamma$, fix $j$ such that $\Cov{X^\beta,X^{\gamma^j}}\neq 0$.
Since for a fixed $j$
$$\frac{c_j(N)}{g_{(|\beta|+|\gamma|)\alpha}(N)^2}\overset{N \rightarrow \infty }{\longrightarrow} 1, $$
and since for any $0< t_1,t_2 \le 1$
$$\frac{g_{t_1}(N)\cdot g_{t_2}(N)}{g_{\frac{t_1+t_2}{2}}(N)^2}\overset{N \rightarrow \infty }{\longrightarrow} C(t_1,t_2)$$
for some $C(t_1,t_2)>0$, we obtain the desired result.
(For $t_1=t_2$, $g_{t_1}(N)\cdot g_{t_2}(N)=g_{\frac{t_1+t_2}{2}}(N)^2$, so if $|\beta|=|\gamma|$, $C(\alpha,\beta,\gamma)=1$).

\end{proof}

\begin{cor} \label{cor:crossterms}
Let $\alpha>0$, $k\in \N$, and let $\beta,\gamma$ be multi-indices which satisfy $\iota(\beta)=\iota(\gamma)=0$. Also assume that $|\beta|=k$ and $|\gamma|>k$. Then
$$\frac{\Cov{\sum_{i=1}^{N} V_\alpha^{\beta^i},\sum_{i=1}^{N}V_\alpha^{\gamma^i}}}{g_{2k\alpha}(N)^2}\overset{N\rightarrow \infty}{\longrightarrow}0. $$
\end{cor}

\begin{lemma} \label{clt_lemma}
Let $B$ be a set of multi-indices, such that every $\beta\in B$ satisfies $\iota(\beta)=0$ and $|\beta|=k$ (for some fixed $k\in\N$). Let $\{a_\beta\}_{\beta\in B}$ be a set of coefficients, such that $a_\beta=0$ for all but finitely many $\beta\in B$. Then
$$\frac{1}{g_{2k\alpha}(N)}\sum_{i=1}^N \sum_{\beta\in B} a_\beta\left(V_\alpha^{\beta^i}-\Ex{V_\alpha^{\beta^i}}\right) \overset{d}{\longrightarrow} N\left(0,\sigma^2\right),$$
as $N \rightarrow \infty$, where $\sigma^2\geq 0$.
\end{lemma}

\begin{proof}
We shall apply Proposition \ref{thm:orey} to the random variables
$$Y_i=\sum_{\beta\in B} a_\beta\left(V_\alpha^{\beta^i}-\Ex{V_\alpha^{\beta^i}}\right). $$
One can verify that $Y_i$ are $|\beta|$-dependent uniformly bounded random variables. From Lemma \ref{norm_cov},
\begin{equation}
D_n^2=\Var{Y_1+...+Y_n}=\sum_{\beta,\gamma\in B}\left(\Cov{\sum_{i=1}^{N} V_\alpha^{\beta^i},\sum_{i=1}^{N}V_\alpha^{\gamma^i}} \right)=O(g_{2k\alpha}(N)^2).
\end{equation}
Since $D_n^2\overset{n\rightarrow \infty }{\longrightarrow} \infty$ and $Y_i$ are uniformly bounded, for every $\epsilon>0$, $\Ex{Y_k^2\cdot 1_{\{|Y_k|>\epsilon D_n\}}}=0$ for $n$ large enough, and the first condition of Proposition \ref{thm:orey} holds.
For the second condition, note that
$$\Var{Y_i}=\sum_{\beta,\gamma\in B} a_{\beta}a_{\gamma}\Cov{V_{\alpha}^{\beta_i},V_{\alpha}^{\gamma_i}}, $$
so again using Lemma \ref{norm_cov}, we get
$$\sum_{i=1}^N \Var{Y_i}=\sum_{\beta,\gamma\in B} \left(a_{\beta}a_{\gamma}\sum_{i=1}^N \Cov{V_{\alpha}^{\beta_i},V_{\alpha}^{\gamma_i}}\right)=O\left(g_{2k\alpha}(N)^2\right).$$
and the second condition of Proposition \ref{thm:orey} is fulfilled. Hence, the proposition holds.
\end{proof}

\subsection{Proof of Propositions \ref{propos:odd}, \ref{propos:even} and \ref{propos:perp}}

\begin{proof}[Proof of Proposition \ref{propos:odd}]
We use Remark \ref{shift_rmk} and Lemma \ref{ARV} to decompose and rearrange
\begin{equation} \label{odd_decomp_1}
\begin{split}
& \Tr{H_{\alpha,N}^k}-\Ex{\Tr{H_{\alpha,N}^k}}\\
&\quad=\sum_{\iota(\beta)\in[1,N]} p^{k}\left(\beta\right)V_\alpha^{\beta}-\Ex{\sum_{\iota(\beta)\in[1,N]} p^{k}\left(\beta\right)V_\alpha^{\beta}}+A^k(N)\\
&\quad=\sum_{i=1}^N \left(\sum_{\iota(\beta)=0} p^k(\beta) \left(V_\alpha^{\beta^i}-\Ex{V_\alpha^{\beta^i}}\right)\right)+A^k(N)\\
&\quad=\sum_{i=1}^{N} \left(\sum_{\beta\in B_1} p^k(\beta) \left(V_\alpha^{\beta^i}-\Ex{V_\alpha^{\beta^i}}\right)\right)\\
&\qquad+\sum_{i=1}^{N} \left(\sum_{\beta\in B_1^\mathsf c} p^k(\beta) \left(V_\alpha^{\beta^i}-\Ex{V_\alpha^{\beta^i}}\right)\right)+A^k(N)\\
\end{split}
\end{equation}
where
\begin{itemize}
\item $B_1$ consists of all multi-indices $\beta$ for which $\iota(\beta)=0$ and $\left|\beta\right|=1$,
\item $B_1^\mathsf c$ consists of all multi-indices $\beta$ for which $\iota(\beta)=0$ and $\left|\beta\right|>1$.
\end{itemize}
Note that for any $k\in\N$, there are only finitely many $\beta$ with $\iota(\beta)=0$ and $p^k(\beta)\neq0$, so the above sums are finite. Moreover, $B_1=\left\{\delta\right\}$. We now use Lemmas \ref{ARV}, \ref{norm_cov} and Corollary \ref{cor:crossterms} to write \eqref{odd_decomp_1} as
\begin{equation} \label{odd_decomp_2}
\Tr{H_{\alpha,N}^k}-\Ex{\Tr{H_{\alpha,N}^k}} = \sum_{i=1}^{N} p^k(\delta) \left(V_\alpha^{\delta^i} - \Ex{V_\alpha^{\delta^i}}\right) +R_{\alpha,k,N}
\end{equation}
so that
\begin{equation} \label{eq:varEst1}
\Var{R_{\alpha,k,N}}=o\left(g_{2\alpha}(N)^2\right),
\end{equation}
and
\begin{equation} \label{eq:covEst1}
\Cov{\sum_{i=1}^{N} p^k(\delta) \left(V_\alpha^{\delta^i} - \Ex{V_\alpha^{\delta^i}}\right),R_{\alpha,k,N}}=o\left(g_{2\alpha}(N)^2\right).
\end{equation}

Using Lemma \ref{clt_lemma}, we deduce
\begin{equation} \nonumber
\frac{\sum_{i=1}^{N} p^k(\delta) \left(V_\alpha^{\delta^i} - \Ex{V_\alpha^{\delta^i}}\right)}{g_{2\alpha}(N)} \overset{d}{\longrightarrow} N\left(0,\sigma_k^2\right),
\end{equation}
as $N\rightarrow\infty$, for some $\sigma_k^2\geq 0$. By \eqref{odd_decomp_2}, \eqref{eq:varEst1} and \eqref{eq:covEst1}
$$\frac{\Tr{H_{\alpha,N}^k}-\Ex{\Tr{H_{\alpha,N}^k}}}{g_{2\alpha}(N)} \overset{d}{\longrightarrow} N\left(0,\sigma_k^2\right).$$

For any odd $k,\ell$, we use Lemma \ref{norm_cov} to compute
\begin{equation} \nonumber
\begin{split}
&\underset{N\longrightarrow \infty}{\lim} \Cov{\frac{\Tr{H_{\alpha,N}^k}-\Ex{\Tr{H_{\alpha,N}^k}}}{g_{2\alpha}(N)}, \frac{\Tr{H_{\alpha,N}^\ell}-\Ex{\Tr{H_{\alpha,N}^\ell}}}{g_{2\alpha}(N)}}\\
&\quad =\underset{N\longrightarrow \infty}{\lim} \Cov{\frac{\sum_{i=1}^{N} p^k(\delta) V_\alpha^{\delta^i}}{g_{2\alpha}(N)} , \frac{\sum_{i=1}^{N} p^\ell(\delta) V_\alpha^{\delta^i}}{g_{2\alpha}(N)}}\\
&\quad =p^k(\delta)p^\ell(\delta) \cdot\underset{N\longrightarrow \infty}{\lim} \Cov{\frac{\sum_{i=1}^{N} V_\alpha^{\delta^i}}{g_{2\alpha}(N)} , \frac{\sum_{i=1}^{N} V_\alpha^{\delta^i}}{g_{2\alpha}(N)}}=p^k(\delta)p^\ell(\delta)\eta^2.
\end{split}
\end{equation}
The value of $\sigma_k^2$ is obtained when $\ell=k$.

For any $P(x)=a_1x+\ldots+a_mx^m\in\mathcal O_m$, from (\ref{odd_decomp_2}) we have
\begin{equation} \nonumber
\begin{split}
\Tr{P\left(H_{\alpha,N}\right)}-\Ex{\Tr{P\left(H_{\alpha,N}\right)}} &= \sum_{i=1}^{N} \left(\sum_{k=1}^m a_k p^k(\delta)\right) \left(V_\alpha^{\delta^i}-\Ex{V_\alpha^{\delta^i}}\right)\\
&\quad+ o\left(g_{2\alpha}(N)\right),
\end{split}
\end{equation}
where the $o$ notation should be interpreted in the sense of \eqref{odd_decomp_2}--\eqref{eq:covEst1}.
Once again, Lemma \ref{clt_lemma} and Corollary \ref{cor:crossterms} tells us that
$$\frac{\Tr{P(H_{\alpha,N})}-\Ex{\Tr{P(H_{\alpha,N})}}}{g_{2\alpha}(N)} \overset{d}{\longrightarrow} N\left(0,\sigma_{\mathcal O}(P)^2\right)$$
as $N\rightarrow\infty$. Now Lemma \ref{norm_cov} allows us to recall (\ref{Qdef}) and (\ref{odd_delta}), and prove
$$\sigma_{\mathcal O}(P)^2=\left \langle P,C^{(m)}_{\mathcal O}P \right \rangle=\left|\left \langle P,Q_m \right \rangle \right|^2 \eta^2.$$
\end{proof}

For an even $k$, things get a bit more complicated, since we now have different types of monomials in $\Tr{H_{\alpha,N}^k}$. Monomials of order $2$ equal $p^k(\delta+\delta^s)V_\alpha(i) V_\alpha(i+s)$, and $p^k(\delta+\delta^s)$ depends on both $k$ and $s$. However, the methods used are essentially the same.

\begin{proof}[Proof of Proposition \ref{propos:even}]
For even $k$, we once again use Remark \ref{shift_rmk} and Lemma \ref{ARV}, to decompose and rearrange
\begin{equation} \label{even_decomp_1}
\begin{split}
& \Tr{H_{\alpha,N}^k}-\Ex{\Tr{H_{\alpha,N}^k}}\\
&\quad=\sum_{\iota(\beta)\in[1,N]} p^{k}\left(\beta\right)V_\alpha^{\beta}-\Ex{\sum_{\iota(\beta)\in[1,N]} p^{k}\left(\beta\right)V_\alpha^{\beta}}+A^k(N)\\
&\quad=\sum_{i=1}^N \left(\sum_{\iota(\beta)=0} p^k(\beta) \left(V_\alpha^{\beta^i}-\Ex{V_\alpha^{\beta^i}}\right)\right)+A^k(N)\\
&\quad=\sum_{i=1}^{N} \left(\sum_{\beta\in B_2} p^k(\beta) \left(V_\alpha^{\beta^i}-\Ex{V_\alpha^{\beta^i}}\right)\right)\\
&\qquad+\sum_{i=1}^{N} \left(\sum_{\beta\in B_2^\mathsf c} p^k(\beta) \left(V_\alpha^{\beta^i}-\Ex{V_\alpha^{\beta^i}}\right)\right)+A^k(N)\\
\end{split}
\end{equation}
where
\begin{itemize}
\item $B_2$ consists of all multi-indices $\beta$ for which $\iota(\beta)=0$ and $|\beta|=2$,
\item $B_2^\mathsf c$ consists of all multi-indices $\beta$ for which $\iota(\beta)=0$ and $|\beta|>2$,
\end{itemize}
again noting that $p^k(\beta)\neq 0$ only for finitely many $\beta\in B_2\cup B_2^\mathsf c$. From Lemmas \ref{ARV}, \ref{norm_cov} and Corollary \ref{cor:crossterms} we again deduce
\begin{equation} \label{even_decomp_2}
\begin{split}
\Tr{H_{\alpha,N}^k}-\Ex{\Tr{H_{\alpha,N}^k}} &= \sum_{i=1}^N \left(\sum_{\beta\in B_2} p^k(\beta)\left(V_\alpha^{\beta^i}-\Ex{V_\alpha^{\beta^i}}\right)\right)\\
&\quad +R_{\alpha,k,N},
\end{split}
\end{equation}
so that
\begin{equation} \label{eq:varEst2}
\Var{R_{\alpha,k,N}}=o\left(g_{4\alpha}(N)^2\right),
\end{equation}
and
\begin{equation} \label{eq:covEst2}
\Cov{\sum_{i=1}^N \left(\sum_{\beta\in B_2} p^k(\beta)\left(V_\alpha^{\beta^i}-\Ex{V_\alpha^{\beta^i}}\right)\right),R_{k,\alpha,N}}
=o\left(g_{4\alpha}(N)^2\right).
\end{equation}

As in the previous proof, it follows from \eqref{even_decomp_2}--\eqref{eq:covEst2} and Lemma \ref{clt_lemma} that
\begin{equation} \nonumber
\begin{split}
\frac{\Tr{H_{\alpha,N}^k}-\Ex{\Tr{H_{\alpha,N}^k}}}{g_{4\alpha}(N)} \overset{d}{\longrightarrow} N\left(0,\sigma_k^2\right)
\end{split}
\end{equation}
as $N\rightarrow\infty$, for some $\sigma_k^2\geq 0$.

Using Lemma \ref{norm_cov}, we compute for even $k,\ell$
\begin{equation} \label{even_cov_1}
\begin{split}
&\underset{N\longrightarrow \infty}{\lim} \Cov{\frac{\Tr{H_{\alpha,N}^k}-\Ex{\Tr{H_{\alpha,N}^k}}}{g_{4\alpha}(N)}, \frac{\Tr{H_{\alpha,N}^\ell}-\Ex{\Tr{H_{\alpha,N}^\ell}}}{g_{4\alpha}(N)}}\\
&\quad =\underset{N\longrightarrow \infty}{\lim} \Cov{\frac{\sum_{i=1}^N \left(\sum_{\beta\in B_2} p^k(\beta) V_\alpha^{\beta^i} \right)}{g_{4\alpha}(N)} , \frac{\sum_{i=1}^N \left(\sum_{\beta\in B_2} p^\ell(\beta) V_\alpha^{\beta^i} \right)}{g_{4\alpha}(N)}}\\
&\quad =\sum_{\beta,\gamma\in B_2} p^k(\beta)p^\ell(\gamma) \cdot\underset{N\longrightarrow \infty}{\lim} \Cov{\frac{\sum_{i=1}^N V_\alpha^{\beta^i}}{g_{4\alpha}(N)} , \frac{\sum_{i=1}^N V_\alpha^{\gamma^i}}{g_{4\alpha}(N)}}\\
&\quad =\sum_{\beta,\gamma\in B_2} p^k(\beta)p^\ell(\gamma)\sum_{j=-\infty}^\infty \Cov{X^\beta,X^{\gamma^j}} .
\end{split}
\end{equation}
Note that any $\beta\in B_2$ is of the form $\delta+\delta^s$, for some $s\geq 0$. Thus for any $\beta,\gamma\in B_2$, we have
$$\Cov{X^\beta,X^{\gamma^j}} =
\begin{cases}
\Ex{X_1^4}-\eta^4 & \beta=\gamma=2\delta\ ,\ j=0 \\
\eta^4 & \beta=\gamma\neq 2\delta\ ,\ j=0\\
0 & \text{otherwise}
\end{cases}.$$
(\ref{even_cov_1}) now becomes
\begin{equation} \label{even_cov_2}
\begin{split}
&\underset{N\longrightarrow \infty}{\lim} \Cov{\frac{\Tr{H_{\alpha,N}^k}-\Ex{\Tr{H_{\alpha,N}^k}}}{g_{4\alpha}(N)}, \frac{\Tr{H_{\alpha,N}^\ell}-\Ex{\Tr{H_{\alpha,N}^\ell}}}{g_{4\alpha}(N)}}\\
&\quad =p^k(2\delta)p^\ell(2\delta)\left(\Ex{X_1^4}-\eta^4\right) + \left(\sum_{s=1}^\infty p^k(\delta+\delta^s)p^\ell(\delta+\delta^s)\right)\eta^4
\end{split}
\end{equation}
as required. Note that the sum in (\ref{even_cov_2}) is in fact finite, as $p^k(\delta+\delta^s)=0$ if $s>(k-2)/2$. The value of $\sigma_k^2$ is obtained when $\ell=k$.

For any $P(x)=a_0+a_2x^2+\ldots+a_{m-1}x^{m-1}\in\mathcal E_m$, from (\ref{even_decomp_2}) we have
\begin{equation} \label{even_poly_decomp}
\begin{split}
&\Tr{P(H_{\alpha,N})}-\Ex{\Tr{P(H_{\alpha,N})}} \\
&\quad=\sum_{i=1}^N \left(\sum_{\beta\in B_2} \left(\sum_{k=0}^{m-1} a_kp^k(\beta)\right) \left(V_\alpha^{\beta^i}-\Ex{V_\alpha^{\beta^i}}\right)\right) + o\left(g_{4\alpha}(N)\right)
\end{split}
\end{equation}
where the $o$ notation should be interpreted in the sense of \eqref{even_decomp_2}--\eqref{eq:covEst2}. We apply Lemma \ref{clt_lemma} to obtain
$$\frac{\Tr{P(H_{\alpha,N})}-\Ex{\Tr{P(H_{\alpha,N})}}}{g_{4\alpha}(N)} \overset{d}{\longrightarrow} N\left(0,\sigma_{\mathcal E}(P)^2\right)$$
as $N\rightarrow\infty$. From Lemma \ref{norm_cov} and (\ref{even_cov_2}) above, we deduce
$$\sigma_{\mathcal E}(P)^2=\left \langle P,C^{(m)}_{\mathcal E}P \right \rangle.$$

It remains to establish the conditions under which $\sigma_{\mathcal E}(P)^2>0$. Our method is to isolate a unique multi-index $\widetilde\beta\in B_2$, such that
$$\sum_{i=1}^N \left(\sum_{k=0}^{m-1} a_k p^k\left(\widetilde\beta\right)\right) \left(V_\alpha^{\widetilde\beta^i}-\Ex{V_\alpha^{\widetilde\beta^i}}\right)$$
is uncorrelated with the other summands in (\ref{even_poly_decomp}), and has non-vanishing variance in the limit.

Denote $d=\deg(P)\geq 2$, and let $\widetilde{\beta}=\delta+\delta^{\frac{d}{2}-1}$. Note that $p^k\left(\widetilde\beta\right)$ is $0$ for every $k<d$, but positive for $k=d$, thus $a_d p^d\left(\widetilde\beta\right)\neq 0$. Define
$$Y^1_N=\sum_{i=1}^N a_d p^d\left(\widetilde\beta\right) \left(V_\alpha^{\widetilde\beta^i}-\Ex{V_\alpha^{\widetilde\beta^i}}\right)$$
and
$$Y_N^2=\sum_{i=1}^N \left(\sum_{\widetilde\beta\neq\beta\in B_2} \left(\sum_{k=0}^d a_kp^k(\beta)\right) \left(V_\alpha^{\beta^i}-\Ex{V_\alpha^{\beta^i}}\right)\right).$$
(\ref{even_poly_decomp}) now becomes
$$\Tr{P(H_{\alpha,N})}-\Ex{\Tr{P(H_{\alpha,N})}}=Y^1_N+Y^2_N+o\left(g_{4\alpha}(N)\right).$$
Note that $V_\alpha^{\beta^i}$ and $V_\alpha^{\gamma^j}$ are uncorrelated for different $\beta,\gamma\in B_2$, therefore $Y^1_N$ and $Y^2_N$ are uncorrelated. Using Lemma \ref{norm_cov}, we obtain
\begin{equation*} \label{eq1}
\begin{split}
&\quad \underset{N\longrightarrow \infty}{\lim} \Var{\frac{\Tr{P(H_{\alpha,N})}-\Ex{\Tr{P(H_{\alpha,N})}}}{g_{4\alpha}(N)}}=\underset{N\longrightarrow \infty}{\lim} \Var{\frac{Y^1_N+Y^2_N}{g_{4\alpha}(N)}}\\
&\geq\underset{N\longrightarrow \infty}{\lim} \Var{\frac{Y^1_N}{g_{4\alpha}(N)}}=
\begin{cases}
a_k^2 \left(p^k(\widetilde\beta)\right)^2 \eta^4 & k>2 \\
a_k^2 \left(p^k(\widetilde\beta)\right)^2 \left(\Ex{X_1^4}-\eta^4\right) & k=2
\end{cases}
\end{split}
\end{equation*}
which is positive under our assumptions.
\end{proof}

Finally,
\begin{proof}[Proof of Proposition \ref{propos:perp}]
The important property of any polynomial $P\in \mathcal Q_m^\perp$, is that the coefficient of $V_\alpha(i)$ in $\Tr{P(H_{\alpha,N})}$ is $0$ for any $k\leq i\leq N-k$ (this can be verified for $P_3,P_5,\ldots,P_m$, which form a basis for $\mathcal{Q}_m^\perp$). Thus, when we proceed in a similar manner to the previous propositions, using Remark \ref{shift_rmk}, Corollary \ref{oddeven2}, and Lemma \ref{ARV}, we obtain
\begin{equation} \label{perp_decomp_1}
\begin{split}
& \Tr{P_k(H_{\alpha,N})} - \Ex{\Tr{P_k(H_{\alpha,N})}}\\
&\quad=\sum_{\underset{|\beta|>1}{\iota(\beta)\in[1,N]}} p^{k}\left(\beta\right)V_\alpha^{\beta}-\Ex{\sum_{\underset{|\beta|>1}{\iota(\beta)\in[1,N]}} p^{k}\left(\beta\right)V_\alpha^{\beta}}+A^k(N)-p^k(\delta)A^1(N)\\
&\quad=\sum_{i=1}^N \left(\sum_{\underset{|\beta|>1}{\iota(\beta)=0}} p^k(\beta) \left(V_\alpha^{\beta^i}-\Ex{V_\alpha^{\beta^i}}\right)\right)+A^k(N)-p^k(\delta)A^1(N)\\
&\quad=\sum_{i=1}^{N} \left(\sum_{\beta\in B_3} p^k(\beta) \left(V_\alpha^{\beta^i}-\Ex{V_\alpha^{\beta^i}}\right)\right)\\
&\qquad+\sum_{i=1}^{N} \left(\sum_{\beta\in B_3^\mathsf c} p^k(\beta) \left(V_\alpha^{\beta^i}-\Ex{V_\alpha^{\beta^i}}\right)\right)+A^k(N)-p^k(\delta)A^1(N)
\end{split}
\end{equation}
where
\begin{itemize}
\item $B_3$ consists of all multi-indices $\beta$ for which $\iota(\beta)=0$ and $|\beta|=3$,
\item $B_3^\mathsf c$ consists of all multi-indices $\beta$ for which $\iota(\beta)=0$ and $|\beta|>3$,
\end{itemize}
once again noting that $p^k(\beta)\neq 0$ only for finitely many $\beta\in B_3\cup B_3^\mathsf c$.

As in the previous proofs, we see by Lemmas \ref{ARV}, \ref{norm_cov} and Corollary \ref{cor:crossterms} that the fluctuations coming from $\beta \in B_3$ dominate the other terms and, applying Lemma \ref{clt_lemma}, we get
%Corollary \ref{cor:crossterms} and equation (\ref{perp_decomp_1}) above, we deduce
%\begin{equation} \label{perp_decomp_2}
%\begin{split}
%\Tr{P_k(H_{\alpha,N})}-\Ex{\Tr{P_k(H_{\alpha,N})}} &= \sum_{i=1}^N \left(\sum_{\beta\in B_3} p^k(\beta)\left(V_\alpha^{\beta^i}-\Ex{V_\alpha^{\beta^i}}\right)\right)\\
%&\quad +o\left(g_{6\alpha}(N)\right),
%\end{split}
%\end{equation}
%where once again, the $o$ notation indicates convergence of the variance of the associated random variable.
%After normalizing by $g_{6\alpha}(N)$, we apply Corollary \ref{cor:crossterms} and Lemma \ref{clt_lemma} to obtain
$$\frac{\Tr{H_{\alpha,N}^k}-\Ex{\Tr{H_{\alpha,N}^k}}}{g_{6\alpha}(N)} \overset{d}{\longrightarrow} N\left(0,\sigma_{\mathcal Q^\perp}(P_k)^2\right)$$
as $N\rightarrow\infty$, for some $\sigma_{\mathcal Q^\perp}(P_k)^2\geq 0$.

Using Lemma \ref{norm_cov}, we now compute for odd $k,\ell\geq 3$:
\begin{equation} \label{perp_cov_1}
\begin{split}
&\underset{N\longrightarrow\infty}{\lim}\Cov{\text{\scalebox{0.9}
{$\frac{\Tr{P_k(H_{\alpha,N})}-\Ex{\Tr{P_k(H_{\alpha,N})}}}{g_{6\alpha}(N)},\frac{\Tr{P_\ell(H_{\alpha,N})}-\Ex{\Tr{P_\ell(H_{\alpha,N})}}}{g_{6\alpha}(N)}$}}}\\
&\quad =\underset{N\longrightarrow \infty}{\lim} \Cov{\frac{\sum_{i=1}^N \left(\sum_{\beta\in B_3} p^k(\beta) V_\alpha^{\beta^i} \right)}{g_{6\alpha}(N)} , \frac{\sum_{i=1}^N \left(\sum_{\beta\in B_3} p^\ell(\beta) V_\alpha^{\beta^i} \right)}{g_{6\alpha}(N)}}\\
&\quad =\sum_{\beta,\gamma\in B_3} p^k(\beta)p^\ell(\gamma) \cdot\underset{N\longrightarrow \infty}{\lim} \Cov{\frac{\sum_{i=1}^N V_\alpha^{\beta^i}}{g_{6\alpha}(N)} , \frac{\sum_{i=1}^N V_\alpha^{\gamma^i}}{g_{6\alpha}(N)}}\\
&\quad =\sum_{\beta,\gamma\in B_3} p^k(\beta)p^\ell(\gamma)\sum_{j=-\infty}^\infty \Cov{X^\beta,X^{\gamma^j}}\\
&\quad=\sum_{\beta,\gamma\in B_3}\sum_{j=-\infty}^\infty p^k(\beta)p^\ell(\gamma) \Cov{X^\beta,X^{\gamma^j}}.
\end{split}
\end{equation}
Our problem is thus reduced to understanding the elements in $B_3$.

All multi-indices in $B_3$ are of the form $\delta+\delta^s+\delta^{s+t}$, for some $s, t\geq 0$. We want to determine conditions on $t$ and $s$ that guarantee a non-zero contribution of the corresponding element to the overall covariance calculation, and to compute that contribution. For example, if $\beta=3\delta$ and $\gamma=2\delta+\delta^1$ then $\Cov{X^\beta,X^{\gamma}}=0$.  In fact, $\Cov{X^\beta,X^{\gamma^j}}=\left(\Ex{X^{\beta+\gamma^j}}-\Ex{X^\beta}\Ex{X^{\gamma^j}}\right)$ is non-zero (for a general distribution of $X_1$) if and only if the supports of $\beta$ and $\gamma^j$ are not disjoint \textit{and} $\beta_n+\gamma^j_n\neq 1$ for all $n$. We proceed to compute the individual contributions of all possible elements to the covariance.

\begin{enumerate}
\item[(a)] From $\beta=\gamma=3\delta$ and $j=0$, we obtain
$$p^k(3\delta)p^\ell(3\delta)\left(\Ex{X_1^6}-\Ex{X_1^3}^2\right).$$

\item[(b)] From $\beta=2\delta+\delta^t$, $\gamma=\delta+2\delta^t$, and $j=0$, we obtain
$$p^k(2\delta+\delta^t) p^\ell(\delta+2\delta^t) \Ex{X_1^3}^2.$$
Symmetrically (from $\beta=\delta+2\delta^t$, $\gamma=2\delta+\delta^t$, and $j=0$), we obtain
$$p^\ell(2\delta+\delta^t) p^k(\delta+2\delta^t) \Ex{X_1^3}^2.$$

\item[(c)] From $\beta=3\delta$, $\gamma=2\delta+\delta^t$, and $j=-t$, we obtain
$$p^k(3\delta)p^\ell(2\delta+\delta^t) \Ex{X_1^4}\Ex{X_1^2}.$$
Symmetrically (from $\beta=2\delta+\delta^t$, $\gamma=3\delta$, and $j=t$), we obtain
$$p^\ell(3\delta)p^k(2\delta+\delta^t) \Ex{X_1^4}\Ex{X_1^2}.$$

\item[(d)] From $\beta=3\delta$, $\gamma=\delta+2\delta^s$, and $j=0$, we obtain
$$p^k(3\delta)p^\ell(\delta+2\delta^s)\Ex{X_1^4}\Ex{X_1^2}.$$
Symmetrically (from $\beta=\delta+2\delta^s$, $\gamma=3\delta$, and $j=0$), we obtain
$$p^\ell(3\delta)p^k(\delta+2\delta^s)\Ex{X_1^4}\Ex{X_1^2}.$$

\item[(e)] From $\beta=\gamma=2\delta+\delta^t$ and $j=0$, we obtain
$$p^k(2\delta+\delta^t)p^\ell(2\delta+\delta^t)\Ex{X_1^4}\Ex{X_1^2}.$$

\item[(f)] From $\beta=\gamma=\delta+2\delta^s$ and $j=0$, we obtain
$$p^k(\delta+2\delta^s) p^\ell(\delta+2\delta^s) \Ex{X_1^4}\Ex{X_1^2}.$$

\item[(g)] From $\beta=2\delta+\delta^{t}$, $\gamma=2\delta+\delta^{s+t}$, and $j=-s$, we obtain
$$p^k(2\delta+\delta^t) p^\ell(2\delta+\delta^{s+t}) \Ex{X_1^2}^3.$$
Symmetrically (from $\beta=2\delta+\delta^{s+t}$, $\gamma=2\delta+\delta^t$, and $j=s$), we obtain
$$p^\ell(2\delta+\delta^t) p^k(2\delta+\delta^{s+t}) \Ex{X_1^2}^3.$$

\item[(h)] From $\beta=2\delta+\delta^t$, $\gamma=\delta+2\delta^s$, and $j=t$, we obtain
$$p^k(2\delta+\delta^t) p^\ell(\delta+2\delta^s) \Ex{X_1^2}^3.$$
Symmetrically (from $\beta=\delta+2\delta^s$, $\gamma=2\delta+\delta^t$, and $j=-t$), we obtain
$$p^\ell(2\delta+\delta^t) p^k(\delta+2\delta^s) \Ex{X_1^2}^3.$$

\item[(i)] From $\beta=\delta+2\delta^s$, $\gamma=\delta+2\delta^{s+t}$, and $j=0$, we obtain
$$p^k(\delta+2\delta^s) p^\ell(\delta+2\delta^{s+t}) \Ex{X_1^2}^3.$$
Symmetrically (from $\beta=\delta+2\delta^{s+t}$, $\gamma=\delta+2\delta^s$, and $j=0$), we obtain
$$p^\ell(\delta+2\delta^s) p^k(\delta+2\delta^{s+t}) \Ex{X_1^2}^3.$$

\item[(j)] Finally, from $\beta=\gamma=\delta+\delta^s+\delta^{s+t}$ and $j=0$, we obtain
$$p^k(\delta+\delta^s+\delta^{s+t}) p^\ell(\delta+\delta^s+\delta^{s+t}) \Ex{X_1^2}^3.$$
\end{enumerate}

Summing the different contributions described above, we see that $(a)$ leads to $q_6$, $(b)$ leads to $q_{3,3}$, $(c)-(f)$ lead to $q_{4,2}$, and $(g)-(j)$ lead to $q_{2,2,2}$. Overall, \eqref{perp_cov_1} becomes
\begin{equation} \label{perp_cov_2}
\begin{split}
&\underset{N\longrightarrow\infty}{\lim}\Cov{\text{\scalebox{0.9}
{$\frac{\Tr{P_k(H_{\alpha,N})}-\Ex{\Tr{P_k(H_{\alpha,N})}}}{g_{6\alpha}(N)},\frac{\Tr{P_\ell(H_{\alpha,N})}-\Ex{\Tr{P_\ell(H_{\alpha,N})}}}{g_{6\alpha}(N)}$}}}\\
&\quad =M_{k,\ell}.
\end{split}
\end{equation}
as required.

Note that if $P(x)=a_1x+a_3x^3+\ldots+a_mx^m\in\mathcal Q_m^\perp$, we also have $P=a_3P_3+a_5P_6+\ldots+a_mP_m$. So as above, we obtain
\begin{equation} \label{perp_poly_decomp}
\begin{split}
&\Tr{P(H_{\alpha,N})}-\Ex{\Tr{P(H_{\alpha,N})}} \\
&\quad=\sum_{i=1}^N \left(\sum_{\beta\in B_3} \left(\sum_{k=3}^m a_kp^k(\beta)\right) \left(V_\alpha^{\beta^i}-\Ex{V_\alpha^{\beta^i}}\right)\right) + o\left(g_{6\alpha}(N)\right),
\end{split}
\end{equation}
where the $o$-notation is to be understood, as in the proofs above, in the sense of the contribution of the corresponding random variable to the asymptotic fluctuations. Lemma \ref{clt_lemma} applies as before, so
$$\frac{\Tr{P(H_{\alpha,N})}-\Ex{\Tr{P(H_{\alpha,N})}}}{g_{6\alpha}(N)} \overset{d}{\longrightarrow} N\left(0,\sigma_{\mathcal Q^\perp}(P)^2\right)$$
as $N\rightarrow\infty$, where
$$\sigma_{\mathcal Q^\perp}(P)^2=\left \langle P,C^{(m)}_{\mathcal Q^\perp}P \right \rangle.$$

It remains to establish the conditions under which $\sigma_{\mathcal Q^\perp}(P)^2>0$. Denote $d=\deg(P)$. As before, we would like to isolate a single $\widetilde\beta\in B_3$, such that
$$\sum_{i=1}^N \left(\sum_{k=3}^d a_kp^k\left(\widetilde\beta\right)\right) \left(V_\alpha^{\widetilde\beta^i}-\Ex{V_\alpha^{\widetilde\beta^i}}\right)$$
is uncorrelated with the other summands in (\ref{perp_poly_decomp}), and has non-vanishing variance in the limit. This is possible for $d\geq 7$ and $d=3$, but not for $d=5$ (which behaves differently).

\subsubsection*{First case, $d\geq 7$:}
Define $\widetilde\beta=\delta+\delta^1+\delta^{\frac{d-3}{2}}$. Note that $p^k\left(\widetilde\beta\right)$ is $0$ for every $k<d$, but positive for $k=d$. In particular, $a_d p^d\left(\widetilde\beta\right)\neq 0$. Define
$$Y^1_N=\sum_{i=1}^N a_dp^d\left(\widetilde\beta\right) \left(V_\alpha^{\widetilde\beta^i}-\Ex{V_\alpha^{\widetilde\beta^i}}\right)$$
and
$$Y^2_N=\sum_{i=1}^N \left(\sum_{\widetilde\beta\neq\beta\in B_3} \left(\sum_{k=3}^d a_kp^k(\beta)\right) \left(V_\alpha^{\beta^i}-\Ex{V_\alpha^{\beta^i}}\right)\right).$$
(\ref{perp_poly_decomp}) now becomes
$$\Tr{P(H_{\alpha,N})}-\Ex{\Tr{P(H_{\alpha,N})}}=Y^1_N+Y^2_N+o\left(g_{6\alpha}(N)\right).$$
Since $X_i X_{i+1} X_{i+\frac{d-3}{2}}$ and $X_j X_{j+s} X_{j+s+t}$ are uncorrelated unless $i=j$, $s=1$ and $t=\frac{d-3}{2}-1$, we deduce that $V_\alpha^{\widetilde\beta^i}$ and $V_\alpha^{\gamma^j}$ are uncorrelated unless $\gamma=\widetilde\beta$ and $j=i$. Therefore, $Y^1_N$ and $Y^2_N$ are uncorrelated, and applying Lemma \ref{norm_cov} we obtain
\begin{equation*}
\begin{split}
&\quad \underset{N\longrightarrow \infty}{\lim} \Var{\frac{\Tr{P(H_{\alpha,N})}-\Ex{\Tr{P(H_{\alpha,N})}}}{g_{6\alpha}(N)}}=\underset{N\longrightarrow \infty}{\lim} \Var{\frac{Y^1_N+Y^2_N}{g_{6\alpha}(N)}}\\
&\geq\underset{N\longrightarrow \infty}{\lim} \Var{\frac{Y^1_N}{g_{6\alpha}(N)}}=a_k^2 \left(p^k(\widetilde\beta)\right)^2 \eta^6>0.
\end{split}
\end{equation*}

\subsubsection*{Second case, $d=3$:}
Since $\mathcal Q_3^\perp$ is spanned by $P_3(x)=x^3-6x$, it suffices to prove that $\sigma_{\mathcal Q^\perp}(P_3)^2>0$. Note that $\beta=3\delta$ is the only multi-index in $B_3$ with $p^3(\beta)\neq 0$ (in fact $p^3(3\delta)=1$), so (\ref{eq:M_kl}) reduces to $\sigma_{\mathcal Q^\perp}(P_3)^2=M_{3,3}=\Ex{X_1^6}-\Ex{X_1^3}^2>0$.

\subsubsection*{Third case, $d=5$:}
Consider the polynomials %$f(x)=P_3(x)=x^3-6x$, and $g(x)=\frac15 P_5(x)-2P_3(x)=\frac{x^5}{5}-2x^3+6x$,
$$f(x)=P_3(x)=x^3-6x\quad,\quad g(x)=\frac15 P_5(x)-2P_3(x)=\frac{x^5}{5}-2x^3+6x$$
which form a basis for $Q_5^\perp$. Using the covariance formulas computed above, we deduce that the random vector
$$\left(\frac{\Tr{f(H_{\alpha,N})}-\Ex{\Tr{f(H_{\alpha,N})}}}{g_{6\alpha}(N)} , \frac{\Tr{g(H_{\alpha,N})}-\Ex{\Tr{g(H_{\alpha,N})}}}{g_{6\alpha}(N)}\right)$$
converges in distribution to a Gaussian vector, with the covariance matrix
$$C\equiv \left( {\begin{array}{cc}
   \Ex{X_1^6}-\Ex{X_1^3}^2 & 2\Ex{X_1^4}\Ex{X_1^2} \\
   2\Ex{X_1^4}\Ex{X_1^2} & 2\left(\Ex{X_1^4}\Ex{X_1^2} + \Ex{X_1^3}^2+\Ex{X_1^2}^3\right) \\
  \end{array} } \right).$$
Next, consider the random variables
$$W_1=\frac{1}{\sqrt 3}\left(X_1^3+X_2^3+X_3^3\right)$$
and
$$W_2=\frac{1}{\sqrt 3}\left(X_1X_2^2+X_2X_3^2+X_3X_1^2 + X_1^2X_2+X_2^2X_3+X_3^2X_1\right).$$
A straightforward computation reveals that the covariance matrix of $(W_1,W_2)$ is the same $C$. Therefore, for any $P\in Q_5^\perp$, find $a,b$ such that $P=a\cdot f+b\cdot g$, and deduce
$$\sigma_{\mathcal Q^\perp}(P)^2=\Var{aW_1+bW_2}.$$
If $\Var{X_1^2}>0$, we may find $A,B$ in the support of $X_1$, such that $B\neq A,-A$. Replacing $(X_1,X_2,X_3)$ with $(A,A,A)\ ,\ (A,A,B)\ ,\ (A,B,B)$ and $(B,B,B)$, and then computing the value of $aW_1+bW_2$, we deduce that unless $a=b=0$, $aW_1+bW_2$ is not supported on a single value, therefore $\Var{aW_1+bW_2}>0$.

However, if $\Var{X_1^2}=0$, there exists some constant $A$, such that $X_1^2=X_2^2=X_3^2=A^2$ with probability $1$. Therefore $W_2=2 W_1$ with probability $1$, and we deduce that the polynomial
$$P(x)=g(x)-2f(x)=\frac{x^5}{5}-4x^3+18x$$
satisfies $\sigma_{\mathcal Q^\perp}(P)^2=0$.
\end{proof}

\end{document}